\documentclass[twocolumn,10pt]{asme2ej}
\usepackage{cite}
\usepackage{amsmath}
\usepackage{amssymb}
\usepackage{amsfonts}
\usepackage{bm}
\usepackage{algorithm}
\usepackage{algpseudocode}

\usepackage{array}
\usepackage[colorlinks]{hyperref}
\hypersetup{
    colorlinks=true,
    linkcolor=black,
    filecolor=magenta,
    citecolor=black,      
    }
\newtheorem{assumption}{Assumption}
\newtheorem{remark}{Remark}
\newtheorem{lemma}{Lemma}
\newtheorem{corollary}{Corollary}
\newtheorem{theorem}{Theorem}

\newtheorem{definition}{Definition}

\usepackage[pdftex]{graphicx}
\usepackage{float}
\usepackage{epstopdf}
\usepackage{xcolor}

\usepackage{enumitem}
\usepackage{subfigure}

\usepackage{mathrsfs}

\usepackage{amsbsy}

\makeatletter
\g@addto@macro\normalsize{%
  \setlength{\abovedisplayskip}{10.0pt plus 2.0pt minus 11.0pt}%
  \setlength{\belowdisplayskip}{10.0pt plus 2.0pt minus 11.0pt}%
  \setlength{\abovedisplayshortskip}{10.0pt plus 2.0pt minus 10.0pt}%
  \setlength{\belowdisplayshortskip}{10.0pt plus 2.0pt minus 10.0pt}%
  \setlength{\intextsep}{5pt}
}
\makeatother
\usepackage[font=small]{caption}

\title{Set-Membership Filtering-Based Leader-Follower Synchronization of Discrete-time Linear Multi-Agent Systems}

\author{Diganta Bhattacharjee \affiliation{PhD student, Department of Mechanical and Aerospace Engineering, \\ The University of Texas at Arlington, Arlington, TX, 76019, USA. \\ email: diganta.bhattacharjee@mavs.uta.edu} % stops a space
}

\author{Kamesh Subbarao \affiliation{Professor, Department of Mechanical and Aerospace Engineering, \\ The University of Texas at Arlington, Arlington, TX, 76019, USA. \\ email: subbarao@uta.edu} % stops a space
}

\begin{document}

\maketitle

\begin{abstract}
{\it In this paper, a set-membership filtering-based leader-follower synchronization protocol for discrete-time linear multi-agent systems is proposed wherein the aim is to make the agents synchronize with a leader. The agents, governed by identical high-order discrete-time linear dynamics, are subject to unknown-but-bounded input disturbances. In terms of its own state information, each agent only has access to measured outputs that are corrupted with unknown-but-bounded output disturbances. Also, the initial states of the agents are unknown. To deal with all these unknowns (or uncertainties), a set-membership filter (or state estimator), having the `correction-prediction' form of a standard Kalman filter, is formulated. We consider each agent to be equipped with this filter that estimates the state of the agent and consider the agents to be able to share the state estimate information with the neighbors locally. The corrected state estimates of the agents are utilized in the local control law design for synchronization. Under appropriate conditions, the global disagreement error between the agents and the leader is shown to be bounded. An upper bound on the norm of the global disagreement error is calculated and shown to be monotonically decreasing. Finally, two simulation examples are included to illustrate the effectiveness of the proposed set-membership filter and the proposed leader-follower synchronization protocol. 
}
\end{abstract}


\vspace{-0.5cm}
\section{Introduction}
%--------------------------------------------------------------------------------------------------------------------------------------------
%                                                    Synchronization of multi-agent systems
%---------------------------------------------------------------------------------------------------------------------------------------------
Cooperative control of multi-agent systems can be applied to solve a number of engineering problems and has attracted much attention in the last few decades. The applications of cooperative control include distributed task assignment and consensus problems, formation flight of spacecrafts and aerial vehicles, distributed estimation problems and so on (see, for example, Refs. \cite{Ren_Beard_2008, Fax_Murray_2004_TAC, Murray_2007, Li_et_al_2009, Wu_et_al_IFAC_2014, Rajnish_ACC_2019}). All of these applications typically require some degree of cooperation and synchronization among the agents. In the context of synchronization (or consensus), there are several types of problems that have been investigated in the existing literature. These are (a) synchronization without a leader (see, for example, Refs. \cite{Trentelman_et_al_2013_TAC, Wang_et_al_2015_TAC}), (b) leader-follower synchronization (see, for example, Refs. \cite{Back_Kim_2017_TAC, Li_et_al_2014, Lewis_et_al_2015_TAC, Arabi_et_al_2017}), (c) average consensus (see, for example, Refs. \cite{Silvestre_et_al_2013_ACC, Silvestre_et_al_2014_ACC}), (d) bipartite consensus (see, for example, Ref. \cite{Valcher_Mishra_2014}) and so on. In this paper, we focus on leader-follower synchronization in the presence of a leader that pins to a group of agents, all having high-order discrete-time linear (time-invariant) dynamics.
%%%%%%%%%%%%%%%%%%%%%%%%%%%%%%%%%%%%%%%%%%%%%%%%%%%%%%%%%%%%%%%%%%%%%%%%%%%%%%%%%%%%%%%%%%%%%%%%%%%%%%%%%%%%%%%%%%%%%%%%%%%%%%%%%%%%%%%%%%%%%%
\vspace{-0.5cm}
\subsection{Motivation}
Most of the studies in the existing literature regarding multi-agent synchronization assume perfect modeling of the system, i.e., the mathematical description of the system is assumed to be perfect (see, for example, Refs. \cite{Wang_et_al_2015_TAC, Li_et_al_2014, Lewis_et_al_2015_TAC, Hengster-Movric_et_al_2013_Automatica, Zhang_et_al_2011_TAC, Arabi_et_al_2017}). However, this is incompatible with real-world engineering problems where the dynamical system is subject to unknown input disturbances, parametric uncertainties, unmodeled dynamics, etc. Because of these phenomena, the state of the system can not be known precisely and should be considered uncertain. For an overview on synchronization in uncertain multi-agent systems, see Refs. \cite{Rajnish_ACC_2019, Trentelman_et_al_2013_TAC, Back_Kim_2017_TAC, Peng_2013} and references therein. Among the abovementioned phenomena, we focus on unknown input disturbances in this paper. Apart from the perfect system modeling assumption, perfect information regarding the states of the agents (full-state feedback) are assumed to be available for synchronization protocol design in a large number of studies (see, for example, Refs. \cite{Hengster-Movric_et_al_2013_Automatica, Li_et_al_2014, Wang_et_al_2015_TAC}). Again, this assumption does not hold for practical applications where measured outputs (some function of the states), subject to output disturbances, are available. Although observer-based approaches, without considering output disturbances, have been investigated in the literature (see, for example, Refs. \cite{Li_et_al_2009, Zhang_et_al_2011_TAC, Back_Kim_2017_TAC}), a state estimation or filtering-based approach would be more suitable to address the effects of both input and output disturbances in the synchronization problem (see, for example, Ref. \cite{Wu_et_al_IFAC_2014}). 
 
The Kalman filter \cite{Anderson_Moore_1979}, which is one of the most widely-studied stochastic filtering techniques, assumes that the input and output disturbances are Gaussian noises with known statistical properties. However, this assumption is difficult to validate in practice and might not hold for real-world systems. Therefore, it seems more realistic to assume the disturbances to be unknown-but-bounded \cite{Polyak_et_al_2004, Chabane_et_al_IFAC_2014}. This approach leads to the concept of set-membership or set-valued state estimation (or filtering) (see, for example, Refs. \cite{Shamma_Tu_1997, El_Ghaoui_Calafiore_2001, Polyak_et_al_2004, Yang_Li_IEEE_Trans_Auto_Control_2009, Yang_Li_Automatica_2009, Chabane_et_al_IFAC_2014, Wei_et_al_2015}), which is deterministic and more suited to several practical applications \cite{Yang_Li_IEEE_Trans_Auto_Control_2009}. Although the set-theoretic or set-valued concepts for synchronization have been investigated in the existing literature (see, for example, Refs. \cite{Xiao_Wang_2012_SIAM, Munz_et_al_2011_TAC, Garulli_Giannitrapani_2011, Sadikhov_et_al_2017}), studies explicitly utilizing set-membership or set-valued estimation techniques for the multi-agent synchronization problem are relatively rare (see, for instance, Refs. \cite{Silvestre_et_al_2013_ACC, Silvestre_et_al_2014_ACC}), despite the practical significance of this class of estimators/filters. Recently, a leader-follower synchronization protocol using set-membership estimation techniques was put forward in Ref. \cite{Ge_et_al_2016}. The synchronization objective in Ref. \cite{Ge_et_al_2016} was to construct ellipsoids centered at the leader's state trajectory that contained states of the agents. This, however, is different from the concept of conventional leader-follower synchronization where the objective is to make the states of the agents converge to the leader's state trajectory. To the best of the authors' knowledge, set-membership estimation techniques have not been employed for the conventional leader-follower synchronization problem in the existing literature.
%%%%%%%%%%%%%%%%%%%%%%%%%%%%%%%%%%%%%%%%%%%%%%%%%%%%%%%%%%%%%%%%%%%%%%%%%%%%%%%%%%%%%%%%%%%%%%%%%%%%%%%%%%%%%%%%%%%%%%%%%%%%%%%%%%%%%%%%%%%%%%
\vspace{-0.5cm}
\subsection{Technical Approach and Contributions}
The technical approach and contributions of the paper are summarized in the following list.
%%%%%%%%%%%%%%%%%%%%%%%%%%%%%%%%%%%%%%%%%%%%%%%%%%%%%%%%%%%%%%%%%%%%%%%%%%%%%%%%%%%%%%%%%%%%%%%%%%%%%%%%%%%%%%%%%%%%%%%%%%%%%%%%%%%%%
\begin{itemize} [label=$\bullet$]
\item We develop a set-membership estimation-based (conventional) leader-follower synchronization protocol for high-order discrete-time linear multi-agent systems with the agents subject to unknown-but-bounded input and output disturbances. To the best of our knowledge, this is a novel contribution.
%-----------------------------------------------------------------------------------------------------------------------------------
\item Specifically, we focus on the ellipsoidal state estimation problem and adopt the terminology set-membership filter (SMF). Based on the approaches given in Refs. \cite{El_Ghaoui_Calafiore_2001, Yang_Li_Automatica_2009, Yang_Li_IEEE_Trans_Auto_Control_2009, Wei_et_al_2015}, we convert the set estimation problem into a recursive algorithm that requires solutions to two semi-definite programs (SDPs) at each time-step. % Compared to the SMFs in Refs. \cite{Yang_Li_IEEE_Trans_Auto_Control_2009, Yang_Li_Automatica_2009, El_Ghaoui_Calafiore_2001, Wei_et_al_2015}, the proposed SMF has a correction-prediction form like the Kalman filter.
%----------------------------------------------------------------------------------------------------------------------------------
\item We consider each agent to be equipped with the SMF that estimates the state of the agent. Further, we assume that the agents are able to share the state estimate information with the neighbors locally and that information is utilized in the local control synthesis for synchronization. The local controller for each agent is chosen based on an $H_2$ type Riccati-based approach \cite{Hengster-Movric_et_al_2013_Automatica}. We show that the global error system is input-to-state stable (ISS) with respect to the input disturbances and estimation errors. Sufficient conditions for input-to-state stability are provided in terms of the system matrices of the agents, the Riccati design, and the interaction graph. 
\item Further, we calculate an upper bound on the norm of the global disagreement error and show that it decreases monotonically, converging to a limit as time goes to infinity.
\end{itemize}
%---------------------------------------------------------------------------------------------------------------------------------------------
\vspace{-0.2cm}
The rest of this paper is organized as follows. Section \ref{Preliminaries} describes the preliminaries required for the SMF design. % and includes introductory concepts from graph theory. 
The formulation of the SMF is given in Section \ref{Set-Membership Filter Design}. The control input synthesis and related results for synchronization are given in Section \ref{SMF based synchronization}. Finally, Section \ref{Simulation Example} includes the simulation examples and Section \ref{Conclusion} presents the concluding remarks.
%%%%%%%%%%%%%%%%%%%%%%%%%%%%%%%%%%%%%%%%%%%%%%%%%%%%%%%%%%%%%%%%%%%%%%%%%%%%%%%%%%%%%%%%%%%%%%%%%%%%%%%%%%%%%%%%%%%%%%%%%%%%%%%%%%%%%%%%%%%%%%
\vspace{-0.3cm}
\subsubsection*{Notations and definitions}
%\textbf{Notations and definitions:} 
The symbol $\mathbb{Z}_{\star}$ denotes the set of non-negative integers. For a square matrix $\bm{X}$, the notation $\bm{X} > 0$ (respectively, $\bm{X} \geq 0$) means $\bm{X}$ is symmetric and positive definite (respectively, positive semi-definite). Similarly, $\bm{X} < 0$ (respectively, $\bm{X} \leq 0$) means $\bm{X}$ is symmetric and negative definite (respectively, negative semi-definite). Further, $\rho(\bm{X})$ denotes the spectral radius of a square matrix $\bm{X}$. %For any matrix $\bm{Y}$, $\sigma_{\text{max}} (\bm{Y})$ and $\sigma_{\text{min}} (\bm{Y})$ stand for maximum and minimum singular value of $\bm{Y}$, respectively.
For any matrix $\bm{Y}$, $\sigma_{\text{max}} (\bm{Y})$ stands for the maximum singular value of $\bm{Y}$. $C(a,b)$ denotes an open circle of radius $b$ in the complex plane, centered at $a \in \mathbb{R}$. Notations $\text{diag}(\cdot)$, $\bm{I}_n$, $\bm{O}_n$, $\bm{1}_{n}$, and $\bm{0}_{n}$ denote block-diagonal matrices, the $n \times n$ identity matrix, the $n \times n$ null matrix, and the vector of ones and zeros of dimension $n$, respectively. For vectors $\bm{x}_1, \bm{x}_2, \dots, \bm{x}_M$, we have $\text{col} [\bm{x}_1, \bm{x}_2, \dots, \bm{x}_M] = [\bm{x}_1^\text{T} \ \bm{x}_2^\text{T} \ \dots \ \bm{x}_M^\text{T}]^\text{T}$. The symbol $|\cdot|$ denotes standard Euclidean norm for vectors and induced matrix norm for matrices. For any function $\bm{\theta}: \mathbb{Z}_\star \rightarrow \mathbb{R}^n$, we have $||\bm{\theta}|| = \text{sup} \{ |\bm{\theta}_k| : k \in \mathbb{Z}_\star \}$. This is the standard $l_{\infty}$ norm for a bounded $\bm{\theta}$. Ellipsoids are denoted by $\mathcal{E}(\bm{c}, \bm{P}) = \{ \bm{x} \in \mathbb{R}^n :  (\bm{x} - \bm{c})^{\text{T}} \bm{P}^{-1} (\bm{x} - \bm{c}) \leq 1 \}$ where $\bm{c} \in \mathbb{R}^n$ is the center of the ellipsoid and $\bm{P} > 0$ is the shape matrix that characterizes the orientation and `size' of the ellipsoid in $\mathbb{R}^n$. Notations $\text{trace} (\cdot)$, $\text{rank} (\cdot)$ denote trace and rank of a matrix, respectively, and $\otimes$ denotes the Kronecker product. The superscript $\text{T}$ means vector or matrix transpose.
%%%%%%%%%%%%%%%%%%%%%%%%%%%%%%%%%%%%%%%%%%%%%%%%%%%%%%%%%%%%%%%%%%%%%%%%%%%%%%%%%%%%%%%%%%%%%%%%%%%%%%%%%%%%%%%%%%%%%%%%%%%%%%%%%%%%%%%%%%%%
\begin{definition}[\cite{Jiang_Wang_2001_Automatica, Lazar_et_al_2013_TAC}]
A function $\gamma : \mathbb{R}_{\geq 0} \rightarrow \mathbb{R}_{\geq 0}$ is a class $\mathcal{K}$ function if it is continuous, strictly increasing and $\gamma(0) = 0$. A function $\beta: \mathbb{R}_{\geq 0} \times \mathbb{R}_{\geq 0}  \rightarrow \mathbb{R}_{\geq 0}$ is a class $\mathcal{KL}$ function if, for each fixed $t \geq 0$, the function $\beta(\cdot,t)$ is a class $\mathcal{K}$ function and for each fixed $s \geq 0$, the function $\beta(s,\cdot)$ is decreasing and $\beta(s,t) \rightarrow 0$ as $t \rightarrow \infty$.
\end{definition}
%%%%%%%%%%%%%%%%%%%%%%%%%%%%%%%%%%%%%%%%%%%%%%%%%%%%%%%%%%%%%%%%%%%%%%%%%%%%%%%%%%%%%%%%%%%%%%%%%%%%%%%%%%%%%%%%%%%%%%%%%%%%%%%%%%%%%%%%%%%%%%%%%%%%%%%%%%%%%%%%%%%%%%%%%%%%%%%%%%%%%%%%%%%%%%%%%%%%%%%%%%%%%%%%%%%%%%%%%%%%%%%%%%%%%%%%%%%%%%%%%%%%%%%%%%%%%%%%%%%%%%%%%%%%%%%%%%%%%%%%%%%%%%%%%%%%%%%%%%%%%%%%%%%%%%%%%%%%%%%%%%%%%%%%%%%%%%%%%%%%%%%%%%%%%%%%%%%%%%%%%%%%%%%%%%%%%%%%%%%%%%%%%%%%%%%%%%%%%%%%%%%%%%%%%%%%%%%%%%%%%%%%%
\vspace{-0.5cm}
\section{Preliminaries}  \label{Preliminaries}
Consider the discrete-time dynamical systems of the form
\begin{equation} \label{discrete-time linear dynamics}
\begin{split}
\bm{x}_{k+1} &= \bm{A}_k \bm{x}_k + \bm{B}_k \bm{u}_k + \bm{G}_k \bm{w}_k, \\ 
\bm{y}_k &= \bm{C}_k \bm{x}_k + \bm{D}_k \bm{v}_k, \quad k \in \mathbb{Z}_{\star}
\end{split}
\end{equation}
where $\bm{x}_k \in \mathbb{R}^{\bar{n}}$ is the state, $\bm{u}_k \in \mathbb{R}^{\bar{m}}$ is the control input, $\bm{w}_k \in \mathbb{R}^{\bar{w}}$ is the input disturbance, $\bm{y}_k \in \mathbb{R}^{\bar{p}}$ is the measured output, $\bm{v}_k \in \mathbb{R}^{\bar{v}}$ is the output disturbance. Also, $\bm{A}_k$, $\bm{B}_k$, $\bm{G}_k$, $\bm{C}_k$ and $\bm{D}_k$ are system matrices of appropriate dimensions. Following are the assumptions for systems of the form given in Eq. \eqref{discrete-time linear dynamics}.
\begin{assumption} \label{Assumption 1} 
\begin{enumerate}[label= 1.\arabic*]
\item \label{Assumption 1 - initial estimate and ellipsoid} The initial state $\bm{x}_0$ is unknown. However, it satisfies $\bm{x}_0 \in \mathcal{E} (\hat{\bm{x}}_0, \bm{P}_0)$ where $\hat{\bm{x}}_0$ is a given initial estimate and $\bm{P}_0$ is known. %Also, $|\bm{P}_0| \leq p_0$ holds with some $p_0 > 0$.
\item \label{Assumption 1 - process and measurement noise ellipsoids} $\bm{w}_k$ and $\bm{v}_k$ are unknown-but-bounded for all $k \in \mathbb{Z}_{\star}$. Also, $\bm{w}_k \in \mathcal{E}(\bm{0}_{\bar{w}}, \bm{Q}_k)$ and $\bm{v}_k \in \mathcal{E}(\bm{0}_{\bar{v}}, \bm{R}_k)$ for all $k \in \mathbb{Z}_{\star}$ where $\bm{Q}_k$, $\bm{R}_k$ are known.
\end{enumerate}
\end{assumption}
%%%%%%%%%%%%%%%%%%%%%%%%%%%%%%%%%%%%%%%%%%%%%%%%%%%%%%%%%%%%%%%%%%%%%%%%%%%%%%%%%%%%%%%%%%%%%%%%%%%%%%%%%%%%%%%%%%%%%%%%%%%%%%%%%%%%%%%%%%%
We intend to develop an SMF for systems of the form in Eq. \eqref{discrete-time linear dynamics}, having a correction-prediction structure  similar to the Kalman filter variants (see, for example, Ref. \cite{Anderson_Moore_1979}). Note that the SMF design in this paper is motivated by the SMF developed by the authors in Ref. \cite{Bhattacharjee_Subbarao}. The filtering objectives are as follows where the corrected and predicted state estimates at time-step $k$ are denoted by $\hat{\bm{x}}_{k|k}$ and $\hat{\bm{x}}_{k+1|k}$, respectively \cite{Bhattacharjee_Subbarao}.
%%%%%%%%%%%%%%%%%%%%%%%%%%%%%%%%%%%%%%%%%%%%%%%%%%%%%%%%%%%%%%%%%%%%%%%%%%%%%%%%%%%%%%%%%%%%%%%%%%%%%%%%%%%%%%%%%%%%%%%%%%%%%%%%%%%%%%%%%%%
\subsection{Correction Step} 
At each time-step $k \in \mathbb{Z}_{\star}$, upon receiving the measured output $\bm{y}_{k}$ with $\bm{v}_{k} \in \mathcal{E} (\bm{0}_{\bar{v}}, \bm{R}_{k})$ and given $\bm{x}_{k} \in \mathcal{E} (\hat{\bm{x}}_{k|k-1}, \bm{P}_{k|k-1})$, the objective is to find the optimal correction ellipsoid, characterized by $\hat{\bm{x}}_{k|k}$ and $\bm{P}_{k|k}$, such that $\bm{x}_{k} \in \mathcal{E} (\hat{\bm{x}}_{k|k}, \bm{P}_{k|k})$. The corrected state estimate is given by 
\begin{equation} \label{objective-corrected state estimate}
\hat{\bm{x}}_{k|k} = \hat{\bm{x}}_{k|k-1} + \bm{L}_{k} (\bm{y}_{k} - \bm{C}_k \hat{\bm{x}}_{k|k-1})
\end{equation}
where $\bm{L}_k$ is the filter gain. Since $\bm{x}_{k} \in \mathcal{E} (\hat{\bm{x}}_{k|k-1}, \bm{P}_{k|k-1})$, there exists a $\bm{z}_{k|k-1} \in \mathbb{R}^{\bar{n}}$ with $|\bm{z}_{k|k-1}|\leq 1$ such that
\begin{equation} \label{True state-prediction step}
\bm{x}_k = \hat{\bm{x}}_{k|k-1} + \bm{E}_{k|k-1} \bm{z}_{k|k-1}
\end{equation}
where $\bm{E}_{k|k-1}$ is the Cholesky factorization of $\bm{P}_{k|k-1}$, i.e., $\bm{P}_{k|k-1} = \bm{E}_{k|k-1} \bm{E}_{k|k-1}^\textnormal{T}$ \cite{El_Ghaoui_Calafiore_2001, Yang_Li_IEEE_Trans_Auto_Control_2009}.
%%%%%%%%%%%%%%%%%%%%%%%%%%%%%%%%%%%%%%%%%%%%%%%%%%%%%%%%%%%%%%%%%%%%%%%%%%%%%%%%%%%%%%%%%%%%%%%%%%%%%%%%%%%%%%%%%%%%%%%%%%%%%%%%%%%%%%%%%%%
\subsection{Prediction Step}
At each time-step $k \in \mathbb{Z}_{\star}$, given $\bm{x}_{k} \in \mathcal{E} (\hat{\bm{x}}_{k|k}, \bm{P}_{k|k})$ and $\bm{w}_k \in \mathcal{E} (\bm{0}_{\bar{w}}, \bm{Q}_k)$, the objective is to find the optimal prediction ellipsoid, characterized by $\hat{\bm{x}}_{k+1|k}$ and $\bm{P}_{k+1|k}$, such that $\bm{x}_{k+1} \in \mathcal{E} (\hat{\bm{x}}_{k+1|k}, \bm{P}_{k+1|k})$ where the predicted state estimate is given by 
\begin{equation} \label{objective-predicted state estimate}
\hat{\bm{x}}_{k+1|k} = \bm{A}_k \hat{\bm{x}}_{k|k} + \bm{B}_k \bm{u}_k 
\end{equation}
Again, since $\bm{x}_{k} \in \mathcal{E} (\hat{\bm{x}}_{k|k}, \bm{P}_{k|k})$, we have
\begin{equation} \label{True state-correction ellipsoid}
\bm{x}_k = \hat{\bm{x}}_{k|k} + \bm{E}_{k|k} \bm{z}_{k|k}
\end{equation}
where $\bm{P}_{k|k} = \bm{E}_{k|k} \bm{E}_{k|k}^\textnormal{T}$ and $|\bm{z}_{k|k}| \leq 1$. Initialization is provided by $\hat{\bm{x}}_{0|-1} = \hat{\bm{x}}_0$ and $\bm{P}_{0|-1} = \bm{P}_0$ \cite{Anderson_Moore_1979}.
%%%%%%%%%%%%%%%%%%%%%%%%%%%%%%%%%%%%%%%%%%%%%%%%%%%%%%%%%%%%%%%%%%%%%%%%%%%%%%%%%%%%%%%%%%%%%%%%%%%%%%%%%%%%%%%%%%%%%%%%%%%%%%%%%%%%%%%
\begin{remark}
As mentioned in the filtering objectives, we are interested in finding the optimal ellipsoids, i.e., the minimum-`size' ellipsoids, at each time-step. There are two criteria for the `size' of an ellipsoid in terms of its shape matrix: trace criterion and log-determinant criterion \cite{El_Ghaoui_Calafiore_2001}. In this paper, we have considered the trace criterion (see Theorems \ref{Theorem: Correction step} and \ref{Theorem: Prediction step}) which represents the sum of squared lengths of semi-axes of an ellipsoid \cite{El_Ghaoui_Calafiore_2001}. As a result, the corresponding optimization problems are convex (see the SDPs in Eqs. \eqref{The complete problem statement-1} and \eqref{The complete problem statement-2}). 
Alternatively, for minimum-volume ellipsoids, one can consider the log-determinant criterion. %which would have computational complexities similar to the trace criterion \cite{El_Ghaoui_Calafiore_2001}. 
However, this would render the optimization problems non-convex and additional modifications might be required to restore convexity (see, for example, Ref. \cite{El_Ghaoui_Calafiore_2001}).
\end{remark}
\section{Set-Membership Filter Design} \label{Set-Membership Filter Design}
In this section, we formulate the SDPs to be solved at each time-step for the SMF. As the SMF design is motivated by the one in \cite{Bhattacharjee_Subbarao}, we have adopted the notations and relevant statements provided in \cite{Bhattacharjee_Subbarao}. First, we state the result that summarizes the filtering problem at the correction step.
%--------------------------------------------------------------------------------------------------------------------------------------------
\begin{theorem}  \label{Theorem: Correction step}
Consider the system in Eq. \eqref{discrete-time linear dynamics} under the Assumption \ref{Assumption 1}. Then, at each time-step $k \in \mathbb{Z}_{\star}$, upon receiving the measured output $\bm{y}_{k}$ with $\bm{v}_{k} \in \mathcal{E} (\bm{0}_{\bar{v}}, \bm{R}_{k})$ and given $\bm{x}_{k} \in \mathcal{E} (\hat{\bm{x}}_{k|k-1}, \bm{P}_{k|k-1})$, the state $\bm{x}_{k}$ is contained in the optimal correction ellipsoid given by $\mathcal{E} (\hat{\bm{x}}_{k|k}, \bm{P}_{k|k})$, if there exist $\bm{P}_{k|k} > 0$, $\bm{L}_{k}$, $\tau_i \geq 0, \ i=1, 2$ as solutions to the following SDP:
\begin{equation} \label{The complete problem statement-1}
\begin{split} 
& \min_{\bm{P}_{k|k}, \bm{L}_{k}, \tau_{1},\tau_{2}} \hspace{0.2cm} \rm{trace}(\bm{P}_{k|k}) \\
& \textnormal{subject to} \\
& \bm{P}_{k|k} > 0 \\
& \tau_i \geq 0, \ i = 1, 2 \\
& \begin{bmatrix}
-\bm{P}_{k|k} & \bm{\Pi}_{k|k-1} \\ \\
%----------------------------------------------------------------------------------------------------------------------------------------------------------------
\bm{\Pi}^\textnormal{T}_{k|k-1} & -\bm{\Theta} (\tau_{1},\tau_{2})
\end{bmatrix} \leq 0
\end{split}
\end{equation} 
where $\bm{\Pi}_{k|k-1}$ and $\bm{\Theta} (\tau_{1},\tau_{2})$ are given by
\begin{equation} \label{Pi_k_k-1 and Theta(tau_1,...tau_6)}
\begin{split}  
\bm{\Pi}_{k|k-1} &= \Big [\bm{0}_{\bar{n}} \quad (\bm{E}_{k|k-1} - \bm{L}_{k} \bm{C}_k \bm{E}_{k|k-1}) \quad - \bm{L}_{k} \bm{D}_k \Big], \\ 
%----------------------------------------------------------------------------------------------------------------------------------------------------------------
\bm{\Theta} (\tau_1, \tau_{2}) &=  \textnormal{diag} \hspace{0.1cm} (1- \tau_1 - \tau_{2}, \tau_{1} \bm{I}_{\bar{n}}, \tau_{2} \bm{R}_{k}^{-1})
\end{split}
\end{equation}
Furthermore, the center of the correction ellipsoid is given by the corrected state estimate in Eq. \eqref{objective-corrected state estimate}.  
\end{theorem}
%%%%%%%%%%%%%%%%%%%%%%%%%%%%%%%%%%%%%%%%%%%%%%%%%%%%%%%%%%%%%%%%%%%%%%%%%%%%%%%%%%%%%%%%%%%%%%%%%%%%%%%%%%%%%%%%%%%%%%%%%%%%%%%%%%%%%%%%%%%%%
\begin{proof}
% \textcolor{blue}{Follows from that of Theorem 1 in \cite{Bhattacharjee_Subbarao}. } \qed
Using Eqs. \eqref{discrete-time linear dynamics}, \eqref{objective-corrected state estimate}, and \eqref{True state-prediction step}, we have
\begin{equation} \label{estimation error at the correction step}
\begin{split}
\hspace{-0.18cm} \bm{x}_{k} - \hat{\bm{x}}_{k|k} &= (\bm{x}_{k} - \hat{\bm{x}}_{k|k-1}) -  \bm{L}_{k} (\bm{y}_{k} - \bm{C}_k \hat{\bm{x}}_{k|k-1})  \\
%----------------------------------------------------------------------------------------------------------------------------------------------------------------------
					 &= (\bm{E}_{k|k-1} - \bm{L}_{k} \bm{C}_k \bm{E}_{k|k-1}) \bm{z}_{k|k-1} - \bm{L}_{k} \bm{D}_k \bm{v}_{k}
\end{split}
\end{equation}
Next, we define $\bm{\zeta} = \text{col} [1, \bm{z}_{k|k-1}, \bm{v}_{k}]$. Therefore, Eq. \eqref{estimation error at the correction step} can be expressed in terms of $\bm{\zeta}$ as
\begin{equation}
\bm{x}_{k} - \hat{\bm{x}}_{k|k} = \bm{\Pi}_{k|k-1} \bm{\zeta}
\end{equation}
where $\bm{\Pi}_{k|k-1}$ is as shown in Eq. \eqref{Pi_k_k-1 and Theta(tau_1,...tau_6)}. Now, $\bm{x}_{k} \in \mathcal{E} (\hat{\bm{x}}_{k|k}, \bm{P}_{k|k})$ is given by 
\begin{equation} \label{State estimation error constraint}
\begin{split}
\bm{\zeta}^{\text{T}}  & \Big[ \bm{\Pi}_{k|k-1}^{\text{T}} \bm{P}_{k|k}^{-1} \bm{\Pi}_{k|k-1} - \text{diag} (1, \bm{O}_{\bar{n}}, \bm{O}_{\bar{v}}) \Big] \bm{\zeta} \leq 0
\end{split} 
\end{equation}
The unknowns in $\bm{\zeta}$ should satisfy the following inequalities: 
\begin{equation}
\begin{split}
\begin{cases}
\bm{z}_{k|k-1}^{\text{T}} \bm{z}_{k|k-1} -1 \leq 0, \\
%--------------------------------------------------
\bm{v}_{k}^\text{T} \bm{R}_k^{-1} \bm{v}_{k}  -1 \leq 0, \\
%------------------------------------------------------
\end{cases}
\end{split}
\end{equation}
which can be expressed in terms of $\bm{\zeta}$ as
\begin{equation}   \label{Constraints on the unknown variables}
\begin{split}
\begin{cases}
\bm{\zeta}^{\text{T}} \text{diag} (-1, \bm{I}_{\bar{n}}, \bm{O}_{\bar{v}}) \bm{\zeta} \leq 0 ,\\
%-------------------------------------------------------------------------------------------------------------------------------------
\bm{\zeta}^{\text{T}} \text{diag} (-1, \bm{O}_{\bar{n}}, \bm{R}_{k}^{-1}) \bm{\zeta} \leq 0 .
\end{cases}
\end{split}
\end{equation}
Next, the S-procedure (see, for example, Ref. \cite{Boyd_et_al_1994}) is applied to the inequalities in Eqs. \eqref{State estimation error constraint} and \eqref{Constraints on the unknown variables}. Thus, a sufficient condition such that the inequalities given in Eq. \eqref{Constraints on the unknown variables} imply the inequality in Eq. \eqref{State estimation error constraint} to hold is that there exist $\tau_1 \geq 0, \tau_2 \geq 0$ such that the following is true:
\begin{displaymath}
\begin{split}
& \bm{\Pi}_{k|k-1}^{\text{T}} \bm{P}_{k|k}^{-1} \bm{\Pi}_{k|k-1} - \text{diag} (1, \bm{O}_{\bar{n}}, \bm{O}_{\bar{v}}) - \tau_1 \text{diag} (-1, \bm{I}_{\bar{n}}, \bm{O}_{\bar{v}}) \\
%-------------------------------------------------------------------------------------------------------------------------------------------------------------------
& - \tau_2 \text{diag} (-1, \bm{O}_{\bar{n}}, \bm{R}_{k}^{-1}) \leq 0
\end{split}
\end{displaymath}
The above inequality can be expressed in a compact form as
\begin{equation} \label{Inequality-compact form}
\begin{split}
\bm{\Pi}_{k|k-1}^{\text{T}} \bm{P}_{k|k}^{-1} \bm{\Pi}_{k|k-1} - \bm{\Theta} (\tau_1,\tau_2) \leq 0
\end{split}
\end{equation}
where $\bm{\Theta} (\tau_1,\tau_2)$ is as shown in Eq. \eqref{Pi_k_k-1 and Theta(tau_1,...tau_6)}. Using the Schur complement (see, for example, Ref. \cite{Boyd_et_al_1994}), we express the inequality in Eq. \eqref{Inequality-compact form} equivalently as
\begin{equation}   \label{Sufficient condition-correction ellipsoid}
\begin{bmatrix}
-\bm{P}_{k|k} &  \bm{\Pi}_{k|k-1} \\ \\
%----------------------------------------------
\bm{\Pi}^\textnormal{T}_{k|k-1} &  -\bm{\Theta} (\tau_1,\tau_2)
\end{bmatrix} \leq 0
\end{equation}
Solving the inequality in Eq. \eqref{Sufficient condition-correction ellipsoid} with $\bm{P}_{k|k} > 0$, $\tau_i \geq 0, \ i = 1, 2 $ yields \textit{a correction ellipsoid} containing the state $\bm{x}_k$. Then, the \textit{optimal correction ellipsoid} is found by minimizing the trace of $\bm{P}_{k|k}$ subject to $\bm{P}_{k|k} > 0$, $\tau_i \geq 0, \ i = 1, 2 $, and Eq. \eqref{Sufficient condition-correction ellipsoid}. This completes the proof. \qed
\end{proof} 
%------------------------------------------------------------------------------------------------------------------------------------- 
Next, we state the technical result for the prediction step.
%----------------------------------------------------------------------------------------------------------------------------------- 
\begin{theorem} \label{Theorem: Prediction step}
Consider the system in Eq. \eqref{discrete-time linear dynamics} under the Assumption \ref{Assumption 1} with the state $\bm{x}_k$ in the correction ellipsoid $\mathcal{E} (\hat{\bm{x}}_{k|k}, \bm{P}_{k|k})$ and $\bm{w}_k \in \mathcal{E} (\bm{0}_{\bar{w}}, \bm{Q}_k)$. Then, the successor state $\bm{x}_{k+1}$ belongs to the optimal prediction ellipsoid $\mathcal{E} (\hat{\bm{x}}_{k+1|k}, \bm{P}_{k+1|k})$, if there exist $\bm{P}_{k+1|k} > 0$, $\tau_i \geq 0, \ i=3,4$ as solutions to the following SDP:
\begin{equation} \label{The complete problem statement-2}
\begin{split} 
& \min_{\bm{P}_{k+1|k}, \tau_3, \tau_4} \hspace{0.2cm} \textnormal{trace}(\bm{P}_{k+1|k}) \\
& \textnormal{subject to} \\
& \bm{P}_{k+1|k} > 0 \\
& \tau_i \geq 0, i = 3,4 \\
& \begin{bmatrix}
-\bm{P}_{k+1|k} &  \bm{\Pi}_{k|k} \\ \\
%----------------------------------------------
\bm{\Pi}^{\textnormal{T}}_{k|k} &  -\bm{\Psi} (\tau_3,\tau_4)
\end{bmatrix} \leq 0 
\end{split}
\end{equation}
where $\bm{\Pi}_{k|k}$ and $\bm{\Psi} (\tau_3,\tau_4)$ are given by
\begin{displaymath}  
\begin{split}
\bm{\Pi}_{k|k} &= \Big [ \bm{0}_{\bar{n}} \quad \bm{A}_k \bm{E}_{k|k} \quad \bm{G}_k \Big], \\ 
%-----------------------------------------------------------------------------------------------------------------------------------------
\bm{\Psi} (\tau_3, \tau_4) &= \textnormal{diag} \ (1- \tau_3 - \tau_4, \tau_3 \bm{I}_{\bar{n}}, \tau_4 \bm{Q}_k^{-1})
\end{split}
\end{displaymath}
Furthermore, the center of the prediction ellipsoid is given by the predicted state estimate in Eq. \eqref{objective-predicted state estimate}.  
\end{theorem}
%-------------------------------------------------------------------------------------------------------------------------------------------------------------
\begin{proof}
% \textcolor{blue}{Follows from that of Theorem 2 in \cite{Bhattacharjee_Subbarao}. }  \qed
Follows directly from the proof of Theorem \ref{Theorem: Correction step} and has been omitted.  \qed
\end{proof}
%---------------------------------------------------------------------------------------------------------------------------------------------
Interior point methods can be implemented to efficiently solve the SDPs in Eqs. \eqref{The complete problem statement-1} and \eqref{The complete problem statement-2} \cite{Vandenberghe_Boyd_1996}. The recursive SMF algorithm is summarized in Algorithm \ref{SMF algorithm}.
%----------------------------------------------------------------------------------------------------------------------------------------------
\begin{algorithm}
\caption{The SMF Algorithm} 
\label{SMF algorithm}
\begin{algorithmic}[1]
\State (Initialization) Select a time-horizon $T_f$. Given the initial values $(\hat{\bm{x}}_0, \bm{P}_0)$, set $k = 0$, $\hat{\bm{x}}_{k|k-1} = \hat{\bm{x}}_0$, $\bm{E}_{k|k-1} = \bm{E}_0$ where $ \bm{P}_0 = \bm{E}_0 \bm{E}_0^{\text{T}}$.
\State Find $\bm{P}_{k|k}$ and $\bm{L}_k$ by solving the SDP in Eq. \eqref{The complete problem statement-1}. 
\State Calculate $\hat{\bm{x}}_{k|k}$ using Eq. \eqref{objective-corrected state estimate}. Also, calculate $\bm{E}_{k|k}$ using $\bm{P}_{k|k} = \bm{E}_{k|k} \bm{E}_{k|k}^\text{T}$.
%\State Using $\hat{\bm{x}}_{k|k}$ and $\bm{E}_{k|k}$, solve the SDP in Eq. \eqref{The complete problem statement-2} to obtain $\bm{P}_{k+1|k}$.
\State Solve the SDP in Eq. \eqref{The complete problem statement-2} to obtain $\bm{P}_{k+1|k}$.
\State Calculate $\hat{\bm{x}}_{k+1|k}$ using Eq. \eqref{objective-predicted state estimate}. Compute $\bm{E}_{k+1|k}$ using $\bm{P}_{k+1|k} = \bm{E}_{k+1|k} \bm{E}_{k+1|k}^\text{T}$. 
\State If $k = T_f$ exit. Otherwise, set $k=k+1$ and go to Step 2. 
\end{algorithmic}
\end{algorithm}
\vspace{-0.5cm}
\section{Leader-Follower Synchronization of Multi-Agent Systems} \label{SMF based synchronization}
This section describes local control input synthesis for the leader-follower synchronization. Results presented in this section are based on the results given in Ref. \cite{Hengster-Movric_et_al_2013_Automatica} and, to be consistent, we have adopted some of the terminologies and notations used in Ref. \cite{Hengster-Movric_et_al_2013_Automatica}. 
%%%%%%%%%%%%%%%%%%%%%%%%%%%%%%%%%%%%%%%%%%%%%%%%%%%%%%%%%%%%%%%%%%%%%%%%%%%%%%%%%%%%%%%%%%%%%%%%%%%%%%%%%%%%%%%%%%%%%%%%%%%%%%%%%%%%%%%%%%%%%%
\vspace{-0.5cm}
\subsection{Graph related preliminaries \cite{Ren_Beard_2008}}
Consider a multi-agent system consisting of $N$ agents. The communication topology of the multi-agent system can be represented by a graph $\mathscr{G} = (\mathscr{V},\mathscr{E})$ where $\mathscr{V} = \{ 1,2,\dots, N \}$ is a nonempty node set and $\mathscr{E} \subseteq \mathscr{V} \times \mathscr{V}$ is an edge set of ordered pairs of nodes, called edges. Node $i$ in the graph represents agent $i$. We consider simple, directed graphs in this paper. The edge $(i,j)$ in the edge set of a directed graph denotes that node $j$ can obtain information from node $i$, but not necessarily \textit{vice versa}. If an edge $(i,j) \in \mathscr{E}$, then node $i$ is a neighbor of node $j$. The set of neighbors of node $i$ is denoted as $\mathscr{N}_i$.

The adjacency matrix $\pmb{\mathscr{A}} = [ a_{ij} ] \in \mathbb{R}^{N\times N}$ of a directed graph $(\mathscr{V}, \mathscr{E})$ is defined such that $a_{ij}$ is a positive weight if $(j,i) \in \mathscr{E}$, and $a_{ij} = 0$ otherwise. The graph Laplacian matrix $\pmb{\mathscr{L}}$ is defined as $\pmb{\mathscr{L}} = \pmb{\mathscr{D}} - \pmb{\mathscr{A}}$ where $\pmb{\mathscr{D}} = [d_{ij}] \in \mathbb{R}^{N\times N}$ is the in-degree matrix with $d_{ij} = 0, i \neq j$, and $d_{ii} = \sum_{j=1}^{N} a_{ij}, i = 1,2,\dots,N $. A directed path is a sequence of edges in a directed graph of the form ($i_1$, $i_2$), ($i_2$, $i_3$), $\dots$. The graph $\mathscr{G}$ contains a (directed) spanning tree if there exists a node, called the root node, such that every other node in $\mathscr{V}$ can be connected by a directed path starting from that node.  
%%%%%%%%%%%%%%%%%%%%%%%%%%%%%%%%%%%%%%%%%%%%%%%%%%%%%%%%%%%%%%%%%%%%%%%%%%%%%%%%%%%%%%%%%%%%%%%%%%%%%%%%%%%%%%%%%%%%%%%%%%%%%%%%%%%%%%%%%%%%%%
\vspace{-0.5cm}
\subsection{Synchronization: Formulation and Results}
We consider $N$ agents connected via a directed graph and a leader. Agent $i$ ($i=1,2, \dots, N$) is a dynamical system of the form 
\begin{equation} \label{i-th agent dynamics}
\begin{split}
\bm{x}^{(i)}_{k+1} &= \bm{A} \bm{x}^{(i)}_{k} + \bm{B} \bm{u}_k^{(i)} + \bm{G} \bm{w}_k^{(i)}, \\ 
\bm{y}^{(i)}_{k}  &= \bm{C} \bm{x}^{(i)}_k + \bm{D} \bm{v}^{(i)}_k, \quad k \in \mathbb{Z}_{\star}
\end{split}
\end{equation}  
where $\bm{x}^{(i)}_{k} \in \mathbb{R}^n$, $\bm{u}_k^{(i)} \in \mathbb{R}^m$, $\bm{y}^{(i)}_{k} \in \mathbb{R}^p$, $\bm{w}_k^{(i)} \in \mathbb{R}^w$, $\bm{v}_k^{(i)} \in \mathbb{R}^v$ are the state, control input, measured output, input and output disturbances for agent $i$, respectively. Clearly, the system described by Eq. \eqref{i-th agent dynamics} is in the form of the system described by Eq. \eqref{discrete-time linear dynamics}, with the time-varying matrices replaced by the constant matrices. Next, we modify Assumption \ref{Assumption 1} and impose the following assumptions on the dynamics of agent $i$ ($i=1,2, \dots, N$).
%%%%%%%%%%%%%%%%%%%%%%%%%%%%%%%%%%%%%%%%%%%%%%%%%%%%%%%%%%%%%%%%%%%%%%%%%%%%%%%%%%%%%%%%%%%%%%%%%%%%%%%%%%%%%%%%%%%%%%%%%%%%%%%%%%%%%%%%%%%%%
\begin{assumption} \label{Assumption 2} 
\begin{enumerate}[label= 2.\arabic*]
\item The initial state $\bm{x}^{(i)}_0$ is unknown. However, it satisfies $\bm{x}^{(i)}_0 \in \mathcal{E} (\hat{\bm{x}}^{(i)}_0, \bm{P}^{(i)}_0)$ where $\hat{\bm{x}}^{(i)}_0$ is a given initial estimate and $\bm{P}^{(i)}_0$ is known. Also, $|\bm{P}^{(i)}_0| \leq p_0$ holds with some $p_0 > 0$.
\item \label{Assumption 2 - process and measurement noise ellipsoids} $\bm{w}^{(i)}_k$ and $\bm{v}^{(i)}_k$ are unknown-but-bounded for all $k \in \mathbb{Z}_{\star}$. Also, $\bm{w}^{(i)}_k \in \mathcal{E}(\bm{0}_{w}, \bm{Q}^{(i)}_k)$ and $\bm{v}^{(i)}_k \in \mathcal{E}(\bm{0}_{v}, \bm{R}^{(i)}_k)$ for all $k \in \mathbb{Z}_{\star}$ where $\bm{Q}^{(i)}_k$, $\bm{R}^{(i)}_k$ are known with $|\bm{Q}^{(i)}_k| \leq \bar{q}$ and $|\bm{R}^{(i)}_k| \leq \bar{r}$ for all $k \in \mathbb{Z}_{\star}$ with some $\bar{q}, \bar{r} > 0$. 
\end{enumerate}
\end{assumption}
%%%%%%%%%%%%%%%%%%%%%%%%%%%%%%%%%%%%%%%%%%%%%%%%%%%%%%%%%%%%%%%%%%%%%%%%%%%%%%%%%%%%%%%%%%%%%%%%%%%%%%%%%%%%%%%%%%%%%%%%%%%%%%%%%%%%%%%%%%%%%
Under this assumption, agent $i$ ($i=1,2, \dots, N$) employs the SMF in Algorithm \ref{SMF algorithm} to estimate its own state. Now, we introduce the following assumption on the system matrices of the agents.
%%%%%%%%%%%%%%%%%%%%%%%%%%%%%%%%%%%%%%%%%%%%%%%%%%%%%%%%%%%%%%%%%%%%%%%%%%%%%%%%%%%%%%%%%%%%%%%%%%%%%%%%%%%%%%%%%%%%%%%%%%%%%%%%%%%%%%%%%%%%%
\begin{assumption} \label{Assumptions on A and B matrices}
$\bm{B}$ is full column rank with the pair $(\bm{A},\bm{B})$ stabilizable.
\end{assumption}
%%%%%%%%%%%%%%%%%%%%%%%%%%%%%%%%%%%%%%%%%%%%%%%%%%%%%%%%%%%%%%%%%%%%%%%%%%%%%%%%%%%%%%%%%%%%%%%%%%%%%%%%%%%%%%%%%%%%%%%%%%%%%%%%%%%%%%%%%%%%%
We consider the leader to be a system of the form
\begin{equation} 
\begin{split}
\bm{x}^{(0)}_{k+1} = \bm{A} \bm{x}^{(0)}_{k}, \
\bm{y}^{(0)}_k = \bm{x}^{(0)}_{k}, \quad k \in \mathbb{Z}_{\star}
\end{split}
\end{equation}
where $\bm{x}^{(0)}_{k} \in \mathbb{R}^n$ is the leader's state and $\bm{y}^{(0)}_k$ is the output. Note that the leader is a virtual system that generates the reference trajectory for the agents $i=1,2, \dots, N$ to track.
%%%%%%%%%%%%%%%%%%%%%%%%%%%%%%%%%%%%%%%%%%%%%%%%%%%%%%%%%%%%%%%%%%%%%%%%%%%%%%%%%%%%%%%%%%%%%%%%%%%%%%%%%%%%%%%%%%%%%%%%%%%%%%%%%%%%%%%%%%%%%
%\begin{remark}
%We assume $\bm{u}^{(0)}_k = \bm{0}_m$, $\bm{w}_k^{(0)} = \bm{0}_w$, $\bm{v}_k^{(0)} = \bm{0}_v$, and $\bm{C}=\bm{I}_n$ for the leader. Thus, the leader is an unforced system having the exact information of its state. These are standard assumptions utilized in the existing literature (see, for example, \cite{Lewis_et_al_2015_TAC, Zhang_et_al_2011_TAC, Li_et_al_2014, Hengster-Movric_et_al_2013_Automatica, Peng_2013, Arabi_et_al_2017}).
%\end{remark}
%%%%%%%%%%%%%%%%%%%%%%%%%%%%%%%%%%%%%%%%%%%%%%%%%%%%%%%%%%%%%%%%%%%%%%%%%%%%%%%%%%%%%%%%%%%%%%%%%%%%%%%%%%%%%%%%%%%%%%%%%%%%%%%%%%%%%%%%%%%%%%
We define the local neighborhood tracking errors, using the corrected state estimates of the agents, as 
\begin{equation*}
\bm{\epsilon}^{(i)}_k = \sum_{j \in \mathscr{N}_i} a_{ij} (\hat{\bm{x}}^{(j)}_{k|k} - \hat{\bm{x}}^{(i)}_{k|k}) + g_i (\bm{x}^{(0)}_{k} - \hat{\bm{x}}^{(i)}_{k|k}) %\ \forall k \in \mathbb{Z}_{\star}
\end{equation*}
where $g_i \geq 0$ are the pinning gains, $\hat{\bm{x}}^{(i)}_{k|k}$ and $\hat{\bm{x}}^{(j)}_{k|k}$ are the corrected state estimates of agent $i$ and $j$, respectively. If agent $i$ is pinned to the leader, we take $g_i > 0$. Now, we choose the control input of agent $i$ as \cite{Hengster-Movric_et_al_2013_Automatica}
\begin{equation*}
\bm{u}^{(i)}_k = c (1+d_{ii} + g_i)^{-1} \bm{K} \bm{\epsilon}^{(i)}_k % \ \forall k \in \mathbb{Z}_{\star}
\end{equation*}
where $c > 0$ is a coupling gain and $\bm{K}$ is a control gain matrix to be discussed subsequently. Hence, the global dynamics of the $N$ agents can be expressed as
\begin{equation} \label{global dynamics of the agents}
\begin{split}
\bm{x}^{(g)}_{k+1} = (\bm{I}_N \otimes \bm{A}) \bm{x}^{(g)}_k + \bm{u}^{(g)}_k + (\bm{I}_N \otimes \bm{G}) \bm{w}^{(g)}_k, \quad k \in \mathbb{Z}_{\star} 
\end{split}
\end{equation} 
with $\bm{x}^{(g)}_k = \text{col} [\bm{x}^{(1)}_k, \dots, \bm{x}^{(N)}_k]$, $\bm{w}^{(g)}_k = \text{col} [\bm{w}^{(1)}_k, \dots, \bm{w}^{(N)}_k]$, and 
\begin{equation}  \label{consenus control synthesis}
\begin{split}
\bm{u}^{(g)}_k  =& -c(\bm{I}_N + \pmb{\mathscr{D}} + \pmb{\mathcal{G}})^{-1} (\pmb{\mathscr{L}} + \pmb{\mathcal{G}}) \otimes \bm{B} \bm{K} \hat{\bm{x}}^{(g)}_{k|k}  \\
			  & + c(\bm{I}_N + \pmb{\mathscr{D}} + \pmb{\mathcal{G}})^{-1} (\pmb{\mathscr{L}} + \pmb{\mathcal{G}}) \otimes \bm{B} \bm{K} \bar{\bm{x}}^{(0)}_k
\end{split}
\end{equation}
where $\pmb{\mathcal{G}} = \text{diag} (g_1, g_2, ...., g_N)$ is the matrix of pinning gains, $\hat{\bm{x}}^{(g)}_{k|k} = \text{col} [\hat{\bm{x}}^{(1)}_{k|k}, \dots, \hat{\bm{x}}^{(N)}_{k|k}]$, and $\bar{\bm{x}}^{(0)}_k = \left( \bm{1}_N \otimes \bm{x}^{(0)}_k \right)$. Note that the superscript $(g)$ is utilized to denote the global variables. Now, using Eq. \eqref{True state-correction ellipsoid} for each agent's corrected state estimates, we can express $\bm{x}^{(g)}_k$ as
\begin{equation} \label{global true state and corrected state estimates}
\bm{x}^{(g)}_k = \hat{\bm{x}}^{(g)}_{k|k} + \bm{E}^{(g)}_{k|k} \bm{z}^{(g)}_{k|k}
\end{equation}
where $\bm{E}^{(g)}_{k|k} = \text{diag} (\bm{E}^{(1)}_{k|k}, \dots, \bm{E}^{(N)}_{k|k})$, $\bm{z}^{(g)}_{k|k} = \text{col} [\bm{z}^{(1)}_{k|k}, \dots, \bm{z}^{(N)}_{k|k}]$. Note that $\bm{E}^{(i)}_{k|k} \left(\bm{E}^{(i)}_{k|k} \right)^\text{T} = \bm{P}^{(i)}_{k|k}$ where $\bm{P}^{(i)}_{k|k}$ is the correction ellipsoid shape matrix for agent $i$ and $|\bm{z}^{(i)}_{k|k}| \leq 1$ for $i = 1,\dots, N$. Our next assumption is regarding the % nature of the
interaction graph.
%%%%%%%%%%%%%%%%%%%%%%%%%%%%%%%%%%%%%%%%%%%%%%%%%%%%%%%%%%%%%%%%%%%%%%%%%%%%%%%%%%%%%%%%%%%%%%%%%%%%%%%%%%%%%%%%%%%%%%%%%%%%%%%%%%%%%%%%%%%%%%
\begin{assumption}[\cite{Hengster-Movric_et_al_2013_Automatica}] \label{Assumption on the graph structure}
The interaction graph contains a spanning tree with at least one nonzero pinning gain that connects the leader and the root node.
\end{assumption}
%%%%%%%%%%%%%%%%%%%%%%%%%%%%%%%%%%%%%%%%%%%%%%%%%%%%%%%%%%%%%%%%%%%%%%%%%%%%%%%%%%%%%%%%%%%%%%%%%%%%%%%%%%%%%%%%%%%%%%%%%%%%%%%%%%%%%%%%%%%%%%
The global disagreement error \cite{Hengster-Movric_et_al_2013_Automatica} is defined as $\bm{\delta}^{(g)}_k = \bm{x}^{(g)}_k - \bar{\bm{x}}^{(0)}_k$. Utilizing Eqs. \eqref{global dynamics of the agents}-\eqref{global true state and corrected state estimates}, we express the global error system as 
\begin{equation} \label{global disagreement error dynamics}
\bm{\delta}^{(g)}_{k+1} = \bm{A}_c \bm{\delta}^{(g)}_k + \bm{B}_c \bm{E}^{(g)}_{k|k} \bm{z}^{(g)}_{k|k} + (\bm{I}_N \otimes \bm{G}) \bm{w}^{(g)}_k, \quad k \in \mathbb{Z}_{\star}
\end{equation} 
where
\begin{equation}
\begin{split}
\bm{A}_c = \left[ (\bm{I}_N \otimes \bm{A}) - c \bm{\Gamma} \otimes \bm{B} \bm{K} \right], \
\bm{B}_c = c \bm{\Gamma} \otimes \bm{B} \bm{K}
\end{split}
\end{equation}
with $\bm{\Gamma} = (\bm{I}_N + \pmb{\mathscr{D}} + \pmb{\mathcal{G}})^{-1} (\pmb{\mathscr{L}} + \pmb{\mathcal{G}})$. Now, we recall the following technical result from Ref. \cite{Hengster-Movric_et_al_2013_Automatica}.
\begin{lemma}[\cite{Hengster-Movric_et_al_2013_Automatica}] \label{Lemma for Schur A_c}
$\rho(\bm{A}_c) < 1$ iff $\rho (\bm{A} - c \Lambda_i \bm{B} \bm{K}) < 1$ for all the eigenvalues $\Lambda_i, i = 1,2, \dots,N$ of $\bm{\Gamma}$.
\end{lemma}
%----------------------------------------------------------------------------------------------------------------------------------------------
If the matrix $\bm{A}$ is unstable or marginally stable, then Lemma \ref{Lemma for Schur A_c} requires Assumption \ref{Assumption on the graph structure} with the pair $(\bm{A}, \bm{B})$ stabilizable \cite{Hengster-Movric_et_al_2013_Automatica}. Using Theorem 2 in Ref. \cite{Hengster-Movric_et_al_2013_Automatica}, $c$ and $\bm{K}$ are chosen such that $\rho(\bm{A}_c) < 1$. To this end, we state the following result.
%----------------------------------------------------------------------------------------------------------------------------------------------
\begin{lemma}[\cite{Hengster-Movric_et_al_2013_Automatica}] \label{Lemma: Choices of c and K for A_c Schur}
Let Assumption \ref{Assumption on the graph structure} hold and let $\pmb{\mathcal{P}}$ be a positive definite solution to the discrete-time Riccati-like equation 
\begin{equation} \label{Riccati-like equation}
\bm{A}^\text{T} \pmb{\mathcal{P}} \bm{A} - \pmb{\mathcal{P}} + \pmb{\mathcal{Q}} - \bm{A}^\text{T} \pmb{\mathcal{P}} \bm{B} (\bm{B}^\text{T} \pmb{\mathcal{P}} \bm{B})^{-1} \bm{B}^\text{T} \pmb{\mathcal{P}} \bm{A} = \bm{O}_n
\end{equation}
for some $\pmb{\mathcal{Q}} > 0$. Define 
\begin{equation*}
r = [\sigma_{\text{max}} (\pmb{\mathcal{Q}}^{-0.5} \bm{A}^\text{T} \pmb{\mathcal{P}} \bm{B} (\bm{B}^\text{T} \pmb{\mathcal{P}} \bm{B})^{-1} \bm{B}^\text{T} \pmb{\mathcal{P}} \bm{A} \pmb{\mathcal{Q}}^{-0.5})]^{-0.5}
\end{equation*}
Further, let there exist a $C(c_0,r_0)$ containing all the eigenvalues $\Lambda_i, i = 1,2, \dots,N$ of $\bm{\Gamma}$ such that $(r_0/c_0) < r$. Then, $\rho(\bm{A}_c) < 1$ for $\bm{K} = (\bm{B}^\text{T} \pmb{\mathcal{P}} \bm{B})^{-1} \bm{B}^\text{T} \pmb{\mathcal{P}} \bm{A}$ and $c=(1/c_0)$.
\end{lemma}
If $\bm{B}$ is full column rank, Eq. \eqref{Riccati-like equation} has a positive definite solution $\pmb{\mathcal{P}}$ only if the pair $(\bm{A}, \bm{B})$ is stabilizable \cite{Hengster-Movric_et_al_2013_Automatica}. In this regard, Assumption \ref{Assumptions on A and B matrices} is pertinent. Next, we introduce the notion of input-to-state stability in the following definition. 
%%%%%%%%%%%%%%%%%%%%%%%%%%%%%%%%%%%%%%%%%%%%%%%%%%%%%%%%%%%%%%%%%%%%%%%%%%%%%%%%%%%%%%%%%%%%%%%%%%%%%%%%%%%%%%%%%%%%%%%%%%%%%%%%%%%%%%%%%%%%%
\begin{definition} \label{Definition: ISS}
A discrete-time system of the form $\bm{x}_{k+1} = \bm{\phi} (\bm{x}_k, \bm{u}_{1_{k}}, \bm{u}_{2_{k}}), \ k \in \mathbb{Z}_{\star}$ with $\bm{u}_{1} : \mathbb{Z}_{\star} \rightarrow \mathbb{R}^{m_1}$,  $\bm{u}_{2} : \mathbb{Z}_{\star} \rightarrow \mathbb{R}^{m_2}$, $\bm{\phi} (\bm{0}_n, \bm{0}_{m_{1}}, \bm{0}_{m_{2}}) = \bm{0}_n$ is (globally) input-to-state stable (ISS) if there exist a class $\mathcal{KL}$ function $\beta$ and two class $\mathcal{K}$ functions $\gamma_1,\gamma_2$ such that, for each pair of inputs $\bm{u}_{1} \in l^{m_1}_{\infty}$, $\bm{u}_{2} \in l^{m_2}_{\infty}$ and each $\bm{x}_0 \in \mathbb{R}^n$, it holds that
\begin{equation} \label{ISS definition}
|\bm{x}_k| \leq \beta(|\bm{x}_0|,k) + \gamma_1 (||\bm{u}_{1}||) + \gamma_2 (||\bm{u}_{2}||)
\end{equation}
for each $k \in \mathbb{Z}_\star$.
\end{definition}
%----------------------------------------------------------------------------------------------------------------------------------------------
\begin{remark}
Definition \ref{Definition: ISS} is adopted from Definition 3.1 in Ref. \cite{Jiang_Wang_2001_Automatica} and has been suitably modified for systems with two inputs using Definition IV.3 in Ref. \cite{Lazar_et_al_2013_TAC}.
\end{remark}
%%%%%%%%%%%%%%%%%%%%%%%%%%%%%%%%%%%%%%%%%%%%%%%%%%%%%%%%%%%%%%%%%%%%%%%%%%%%%%%%%%%%%%%%%%%%%%%%%%%%%%%%%%%%%%%%%%%%%%%%%%%%%%%%%%%%%%%%%%%%%
We state the main result of this section in the next theorem.
\begin{theorem}  \label{Theorem: the global disagreement error dynamics is ISS}
Suppose the following conditions are satisfied: (i) Under Assumption \ref{Assumption 2}, agent $i$ ($i=1,2,\dots,N$) employs the SMF in Algorithm \ref{SMF algorithm} to estimate its own state; (ii) Assumptions \ref{Assumptions on A and B matrices} and \ref{Assumption on the graph structure} hold; (iii) $c$ and $\bm{K}$ are chosen using Lemma \ref{Lemma: Choices of c and K for A_c Schur}. Then, the global error system in Eq. \eqref{global disagreement error dynamics} is ISS.
\end{theorem}
\begin{proof}
The proof is inspired by Example 3.4 in Ref. \cite{Jiang_Wang_2001_Automatica}. Let us denote $\bm{e}^{(g)}_k = \text{col} [\bm{e}^{(1)}_k, \dots, \bm{e}^{(N)}_k]$ where $\bm{e}^{(i)}_k = \bm{x}^{(i)}_k - \hat{\bm{x}}^{(i)}_{k|k}$ are the state estimation errors of agent $i$ at the correction steps. Now, using Eq. \eqref{global true state and corrected state estimates}, we have $\bm{e}^{(g)}_k = \bm{x}^{(g)}_k - \hat{\bm{x}}^{(g)}_{k|k} = \bm{E}^{(g)}_{k|k} \bm{z}^{(g)}_{k|k}$. Similarly, let us denote $\bm{e}^{(g)}_0 = \text{col} [\bm{e}^{(1)}_0, \dots, \bm{e}^{(N)}_0]$ where $\bm{e}^{(i)}_0 = \bm{x}^{(i)}_0 - \hat{\bm{x}}^{(i)}_{0}$ is the initial estimation error of agent $i$. Due to Assumption \ref{Assumption 2}, we have 
\begin{equation} \label{Initial global estimation error}
\bm{e}^{(g)}_0 = \bm{E}^{(g)}_{0} \bm{z}^{(g)}_0
\end{equation}
with $\bm{E}^{(g)}_0 = \text{diag} (\bm{E}^{(1)}_0, \dots, \bm{E}^{(N)}_0)$, $\bm{z}^{(g)}_0 = \text{col} [\bm{z}^{(1)}_0, \dots, \bm{z}^{(N)}_0]$ where $\bm{E}^{(i)}_{0} \left(\bm{E}^{(i)}_{0} \right)^\text{T} = \bm{P}^{(i)}_{0}$ and $|\bm{z}^{(i)}_{0}| \leq 1$ for $i = 1,\dots, N$. Then, Eq. \eqref{global disagreement error dynamics} becomes
\begin{equation*}  \label{global disagreement error-modified form}
\bm{\delta}^{(g)}_{k+1} = \bm{A}^{k+1}_c \bm{\delta}^{(g)}_{0} + \sum_{j=0}^{k} \bm{A}_c^j \bm{B}_c \bm{e}^{(g)}_{k-j} +  \sum_{j=0}^{k} \bm{A}_c^j (\bm{I}_N \otimes \bm{G}) \bm{w}^{(g)}_{k-j} 
\end{equation*}
where $\bm{e}^{(g)}: \mathbb{Z}_{\star} \rightarrow \mathbb{R}^{nN},\ \bm{w}^{(g)}: \mathbb{Z}_{\star} \rightarrow \mathbb{R}^{wN}$ are the inputs. It is understood that $\bm{e}^{(g)} \in l^{nN}_{\infty}$ and $\bm{w}^{(g)} \in l^{wN}_{\infty}$. Due to the choices of $c$ and $\bm{K}$ along with Assumptions \ref{Assumptions on A and B matrices} and \ref{Assumption on the graph structure}, we have $\rho(\bm{A}_c) < 1$. Hence, there exist constants $\alpha>0$ and $\mu \in [0,1)$ such that $|\bm{A}_c^k| \leq \alpha \mu^k, \ k \in \mathbb{Z}_{\star}$ \cite{Jiang_Wang_2001_Automatica}. Then, the ISS property in Eq. \eqref{ISS definition} holds for the system in Eq. \eqref{global disagreement error dynamics} with
\begin{equation}  \label{ISS property related functions}
\begin{split}
& \beta(s,k) = \alpha \mu^k s, \ \gamma_1(s_1) = \sum_{j=0}^{\infty} \alpha \mu^j |\bm{B}_c| s_1 = \frac{\alpha |\bm{B}_c| s_1}{1 - \mu}, \\
%-------------------------------------------------------------------------------------------------------------------------------------------
& \gamma_2 (s_2) = \sum_{j=0}^{\infty} \alpha \mu^j |\bm{G}| s_2 = \frac{\alpha |\bm{G}| s_2}{1 - \mu} 
\end{split}
\end{equation} 
where $|(\bm{I}_N \otimes \bm{G})| = |\bm{I}_N| \ |\bm{G}| = |\bm{G}|$ is utilized. Thus, along the trajectories of the system in Eq. \eqref{global disagreement error dynamics}, for each $k \in \mathbb{Z}_\star$, it holds that 
\begin{equation} \label{global error system - ISS property}
|\bm{\delta}_k^{(g)}| \leq \beta(|\bm{\delta}_0^{(g)}|,k) + \gamma_1 (||\bm{e}^{(g)}||) + \gamma_2 (||\bm{w}^{(g)}||)
\end{equation}
where the functions $\beta, \gamma_1, \gamma_2$ are as in Eq. \eqref{ISS property related functions} with $s = |\bm{\delta}_0^{(g)}|$, $s_1 = ||\bm{e}^{(g)}||$, $s_2 = ||\bm{w}^{(g)}||$. This completes the proof. \qed
\end{proof}
Theorem \ref{Theorem: the global disagreement error dynamics is ISS} implies that the global disagreement error remains bounded under the proposed synchronization protocol. 
%%%%%%%%%%%%%%%%%%%%%%%%%%%%%%%%%%%%%%%%%%%%%%%%%%%%%%%%%%%%%%%%%%%%%%%%%%%%%%%%%%%%%%%%%%%%%%%%%%%%%%%%%%%%%%%%%%%%%%%%%%%%%%%%%%%%%%%%%%%
\begin{remark}
Objective of the SMF-based synchronization in Ref. \cite{Ge_et_al_2016} was to contain the states of the agents in a confidence ellipsoid which might not be small in general. Thus, the approach outlined in Ref. \cite{Ge_et_al_2016} may lead to conservative results where the states of the agents might not converge to a neighborhood of the leader's state trajectory. On the other hand, we have shown that, under appropriate conditions, the global error system is ISS with respect to the input disturbances and estimation errors. Since an ISS system admits the `converging-input converging-state' property (see, Refs. \cite{Jiang_Wang_2001_Automatica, Sontag_2003} for details), $|\bm{\delta}^{(g)}_k|$ would eventually converge to a neighborhood of zero as the estimation errors %magnitudes 
of the agents decrease. Thus, the agents would converge to a neighborhood of the leader's state trajectory. To this end, it is understood that $||\bm{w}^{(g)}||$ is relatively small (compared to $|\bm{\delta}^{(g)}_0|$ and $||\bm{e}^{(g)}||$) as the input disturbances satisfy Assumption \ref{Assumption 2 - process and measurement noise ellipsoids}.
\end{remark}
%%%%%%%%%%%%%%%%%%%%%%%%%%%%%%%%%%%%%%%%%%%%%%%%%%%%%%%%%%%%%%%%%%%%%%%%%%%%%%%%%%%%%%%%%%%%%%%%%%%%%%%%%%%%%%%%%%%%%%%%%%%%%%%%%%%%%%%%%%%%
Next, we state the following result based on Theorem \ref{Theorem: the global disagreement error dynamics is ISS} where $p_0$ and $\bar{q}$ are as in Assumption \ref{Assumption 2}.
%%%%%%%%%%%%%%%%%%%%%%%%%%%%%%%%%%%%%%%%%%%%%%%%%%%%%%%%%%%%%%%%%%%%%%%%%%%%%%%%%%%%%%%%%%%%%%%%%%%%%%%%%%%%%%%%%%%%%%%%%%%%%%%%%%%%%%%%%%%%%%
\begin{corollary}
Under the conditions of Theorem \ref{Theorem: the global disagreement error dynamics is ISS}, the normalized global disagreement error $\bar{\bm{\delta}}^{(g)}_k$ satisfies  
\begin{equation} \label{final bound of |delta_k|}
\lim_{k \rightarrow \infty} |\bar{\bm{\delta}}^{(g)}_{k}| \leq (|\bm{B}_c| \sqrt{p_0} + |\bm{G}| \sqrt{\bar{q}})
%-------------------------------------------------------------------------------------------
% \alpha \sqrt{N} (|\bm{B}_c| \sqrt{p_0} + |\bm{G}| \sqrt{\bar{q}}) \left(\frac{1}{1-\mu} \right)
% \alpha \sqrt{N} (c \sqrt{N p_0} |\bm{B}| |\bm{K}| + |\bm{G}| \sqrt{\bar{q}})
\end{equation}
with $\bar{\bm{\delta}}^{(g)}_k = \left( \bm{\delta}^{(g)}_k /\bar{\mu} \right) $ where % $\bar{\mu} = \alpha \sqrt{N} \left(\frac{1}{1-\mu} \right)$
$\bar{\mu} = \left( \alpha \sqrt{N} / (1-\mu) \right)$ and $\alpha > 0$, $\mu \in [0,1)$ are such that $|\bm{A}_c^k| \leq \alpha \mu^k$ for all $k \in \mathbb{Z}_{\star}$.
\end{corollary}
%--------------------------------------------------------------------------------------------------------------------------------------
\begin{proof}
Under the conditions of Theorem \ref{Theorem: the global disagreement error dynamics is ISS}, the result in Eq. \eqref{global error system - ISS property} holds. Then, let us rewrite Eq. \eqref{global error system - ISS property} as
\begin{equation*}
|\bm{\delta}^{(g)}_{k}| \leq \alpha \mu^{k} |\bm{\delta}^{(g)}_{0}| + \left( \alpha / (1-\mu) \right) \left( |\bm{B}_c| \ ||\bm{e}^{(g)}|| + |\bm{G}| \ ||\bm{w}^{(g)}|| \right)
\end{equation*} 
Now, under the assumption that the SMFs of the agents are performing adequately, % and the initial ellipsoids are sufficiently large,
we can utilize Eq. \eqref{Initial global estimation error} and take $||\bm{e}^{(g)}|| \leq |\bm{e}^{(g)}_0| \leq |\bm{E}^{(g)}_{0}| \ |\bm{z}^{(g)}_0|$. Using Assumption \ref{Assumption 2}, we have $|\bm{E}^{(g)}_{0}| = \text{max} (|\bm{E}^{(1)}_0|, \dots, |\bm{E}^{(N)}_0|) \Rightarrow |\bm{E}^{(g)}_{0}| \leq \sqrt{p_0}$. Also, we have $|\bm{z}^{(g)}_0| \leq \sqrt{N}$. Therefore, we derive $||\bm{e}^{(g)}|| \leq \sqrt{p_0 N}$. Similarly, Assumption \ref{Assumption 2} leads to $||\bm{w}^{(g)}|| \leq \sqrt{\bar{q} N}$. Combining these, we calculate the following bound on $\bm{\delta}^{(g)}_{k}$ 
\begin{equation} \label{monotonically decreasing upper bound on the norm of global disagreement error}
|\bm{\delta}^{(g)}_{k}| \leq \alpha \mu^{k} |\bm{\delta}^{(g)}_{0}| + \bar{\mu} \left( |\bm{B}_c| \sqrt{p_0} + |\bm{G}| \sqrt{\bar{q}} \right)
\end{equation}
for each $k \in \mathbb{Z}_{\star}$ with $\bar{\mu} = \left( \alpha \sqrt{N} / (1-\mu) \right)$. % $\bar{\mu} = \alpha \sqrt{N} \left(\frac{1}{1-\mu} \right)$.
Hence, the proof is completed by taking the limit in Eq. \eqref{final bound of |delta_k|} and carrying out the normalization $\bar{\bm{\delta}}^{(g)}_k = \left( \bm{\delta}^{(g)}_k /\bar{\mu} \right)$. \qed % with $\bar{\mu} = \alpha \sqrt{N} \left(\frac{1}{1-\mu} \right)$.   \qed
\end{proof}

%----------------------------------------------------------------------------------------------------------------------------------------
\begin{remark}
The upper bound shown in Eq. \eqref{monotonically decreasing upper bound on the norm of global disagreement error} is monotonically decreasing. The estimate given in Eq. \eqref{final bound of |delta_k|} is a conservative one as we have utilized $||\bm{e}^{(g)}|| \leq \sqrt{p_0 N}$ and $||\bm{w}^{(g)}|| \leq \sqrt{\bar{q} N}$. Also, the bound $|\bm{R}^{(i)}_k| \leq \bar{r}$ does not appear in Eqs. \eqref{final bound of |delta_k|} and \eqref{monotonically decreasing upper bound on the norm of global disagreement error} as a result of utilizing $||\bm{e}^{(g)}|| \leq |\bm{e}^{(g)}_0| \leq |\bm{E}^{(g)}_{0}| \ |\bm{z}^{(g)}_0|$. However, the true value of $\bm{e}^{(i)}_{k}$ would depend on $\bm{v}^{(i)}_k$ and, thus, on $\bm{R}^{(i)}_k$ for all $k \in \mathbb{Z}_\star$ and all $i=1,2,\dots,N$.
\end{remark}
%-------------------------------------------------------------------------------------------------------------------------------------------
\begin{remark}
For a given multi-agent system (with the number of agents $N$, the matrices $\bm{A}, \ \bm{B}, \ \bm{C}, \ \bm{D}, \ \bm{G}$, and the interaction graph specified), we have $|\bm{B}_c|$ and $|\bm{G}|$ fixed once $c$ and $\bm{K}$ are properly chosen using Lemma \ref{Lemma: Choices of c and K for A_c Schur}. Thus, the conservatism of the bound in Eq. \eqref{final bound of |delta_k|} can be reduced if the available upper bounds (i.e., $p_0$ and $\bar{q}$) are sufficiently small. 
\end{remark} 
%%%%%%%%%%%%%%%%%%%%%%%%%%%%%%%%%%%%%%%%%%%%%%%%%%%%%%%%%%%%%%%%%%%%%%%%%%%%%%%%%%%%%%%%%%%%%%%%%%%%%%%%%%%%%%%%%%%%%%%%%%%%%%%%%%%%%%%%%%%%%%%%%%%%%%%%%%%%%%%%%%%%%%%%%%%%%%%%%%%%%%%%%%%%%%%%%%%%%%%%%%%%%%%%%%%%%%%%%%%%%%%%%%%%%%%%%%%%%%%%%%%%%%%%%%%%%%%%%%%%%%%%%%%%%%%%%%%%%%%%%%%%%%%%%%%%%%%%%%%%%%%%%%%%%%%%%%%%%%%%%%%%%%%%%%%%%%%%%%%%%%%%%%%%%%%%%%%%%%%%%%%%%%%%%%%%%%%%%%%%%%%%%%%%%%%%%%%%%%%%%%%%%%%%%%%%%%%%%%%%%%%%%%
\vspace{-0.5cm}
\section{Simulation Examples}  \label{Simulation Example}
Simulation examples are provided in this section to illustrate the effectiveness of the proposed SMF and SMF-based leader-follower synchronization protocol. All the simulations are carried out on a desktop computer with a 16.00 GB RAM and a 3.40 GHz Intel(R) Xeon(R) E-2124 G processor running MATLAB R2019a. The SDPs in Eqs. \eqref{The complete problem statement-1} and \eqref{The complete problem statement-2} are solved utilizing `YALMIP' \cite{Lofberg_2004} with the `SDPT3' solver in the MATLAB framework. Since the disturbances are only assumed to be unknown-but-bounded, different kinds of disturbance realizations are possible that satisfy the ellipsoidal assumptions (Assumptions \ref{Assumption 1 - process and measurement noise ellipsoids} and \ref{Assumption 2 - process and measurement noise ellipsoids}). For example, periodic disturbances with time-varying or constant frequencies and amplitudes, random disturbances with each element being uniformly distributed in an interval, and so on.  
\vspace{-0.25cm}
\subsection{Example-1}
In this example, we illustrate the effectiveness of the proposed SMF algorithm and compare our results with the results obtained for the SMF in Ref. \cite{Balandin_et_al_Automatica_2020} (the discrete version). We choose a system governed by the Mathieu equation \cite{Balandin_et_al_Automatica_2020} for this example, i.e., the system given by
\begin{equation} \label{Mathieu equation}
\begin{split}
\dot{x}_1 &= x_2, \\
\dot{x}_2 &= - \omega_0^2 (1 + \epsilon \sin \omega t) x_1 + w_d
\end{split}
\end{equation}  
which is expressed in a compact form as
\begin{equation}
\dot{\bm{x}} = \begin{bmatrix}
\dot{x}_1 \\ \dot{x}_2
\end{bmatrix} = \bm{A}(t)  \bm{x}  + \bm{G} w_d
\end{equation}
where 
\begin{equation}
\bm{A} (t) = \begin{bmatrix}
0 && 1\\
%-----------------------------------------
- \omega_0^2 (1 + \epsilon \sin \omega t) && 0
\end{bmatrix}, \quad \bm{G} = \begin{bmatrix}
0 \\ 1
\end{bmatrix}
\end{equation}
with $w_d$ as the input disturbance. Utilizing zeroth-order hold (ZOH) with a sampling time $\Delta t$, the above system is put into an equivalent discrete-time form as 
\begin{equation}
\bm{x}_{k+1} = \bm{A}_k \bm{x}_k + \bm{G}_k w_k
\end{equation}
where $\bm{x}_{(\cdot)} = [x_{1_{(\cdot)}} \quad  x_{2_{(\cdot)}} ]^\text{T}$ and $w_k$ is the input disturbance at the current time-step. The measured outputs are considered as $y_k = x_{1_{k}} + v_k$.
%------------------------------------------------------------------------------------------------------------------------------------------
\begin{figure}[!hbt]
\centering
\includegraphics[width=0.9\columnwidth]{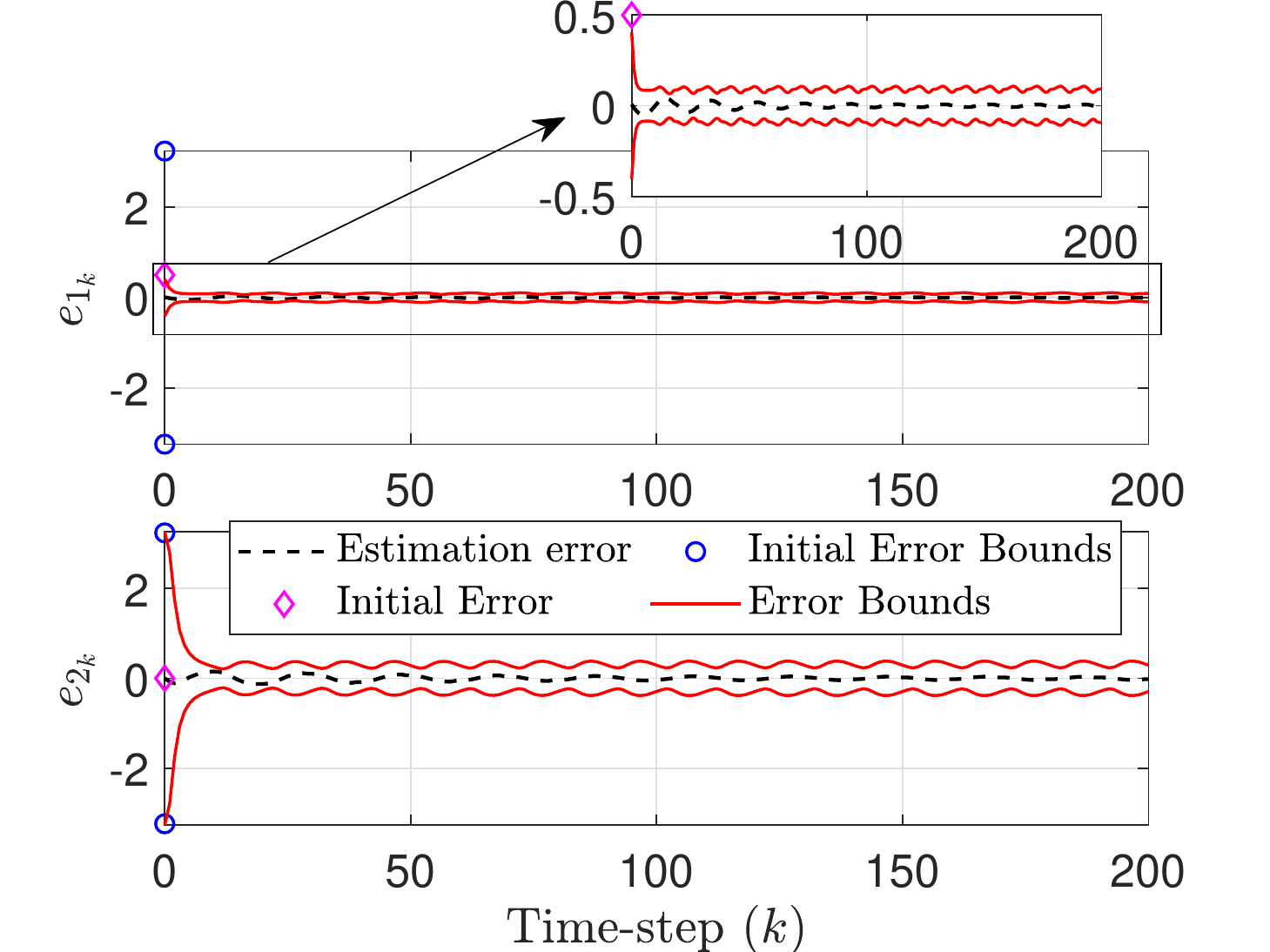}
\caption{Estimation errors %(at the correction steps) 
and corresponding error bounds for the proposed SMF (Example-1).} 
\label{SMF_estimation_error_bounds}
\end{figure}
%-------------------------------------------------------------------------------------------------------------------------------------------
%\begin{equation}
%y_k = x_{1_{k}} + v_k
%\end{equation}
Thus, we have $\bm{C}_k = [1 \quad 0]$, $\bm{D}_k = 1$. The system parameters, $\Delta t$, disturbances, disturbance ellipsoid shape matrices, and initial conditions chosen are as follows:
\begin{equation}  \label{Parameter values for the SMF}
\begin{split}
\omega &= 2 \pi, \ \omega_0 = \pi, \ \epsilon = 0.3, \ \Delta t = 0.1 \ \rm{seconds}  \\
%-----------------------------------------------------------------------------------------------------------
w_k &= 0.05 \sin (\omega t_k), \ v_k = w_k, \ \bm{Q}_k = 0.0025, \ \bm{R}_k = \bm{Q}_k, \\
%-----------------------------------------------------------------------------------------------------------
\bm{x}_0 &= [0.5 \quad 0]^\text{T}, \ \hat{\bm{x}}_0 = \bm{0}_2, \ \bm{P}_0 = 10.5 \bm{I}_2
\end{split}
\end{equation}
With the above initial conditions, disturbances and disturbance ellipsoid shape matrices, Assumption \ref{Assumption 1} is satisfied, and we implement the proposed SMF in Algorithm \ref{SMF algorithm} with $T_f = 200$. The simulation results are shown in Figs. \ref{SMF_estimation_error_bounds} and \ref{SMF_phase_plane}. The estimation errors and error bounds shown in Fig. \ref{SMF_estimation_error_bounds} are corresponding to the correction steps. Thus, we have $\bm{e}_k = \bm{x}_k - \hat{\bm{x}}_{k|k} = [e_{1_{k}} \quad  e_{2_{k}} ]^\text{T} $. %The error bounds shown in Fig. \ref{SMF_estimation_error_bounds} are calculated using $\bm{e}_k \in \mathcal{E} (\hat{\bm{0}}_2, \bm{P}_{k|k})$ (see Remark \ref{Remark on estimation error ellipsoids}).
As shown in Fig. \ref{SMF_estimation_error_bounds}, the estimation errors remain within the corresponding error bounds for the entire time-horizon considered. Thus, the true state is contained in the correction ellipsoids for the entire time-horizon too. The error bounds are time-varying for this example as the dynamical system considered here is time-varying. Further, the error bounds decrease significantly from the corresponding initial values, evidencing the ongoing optimization process for the SMF. The phase plane plot of the true states and corrected state estimates are shown in Fig. \ref{SMF_phase_plane}. Since the estimation errors are small (as shown in Fig. \ref{SMF_estimation_error_bounds}), the true state and corrected state estimate trajectories remain in a close neighborhood of each other. This is depicted in Fig. \ref{SMF_phase_plane}. Also, the zoomed-in plot in Fig. \ref{SMF_phase_plane} shows that the SMF is able to bring the corrected state estimate in a neighborhood of the true state after the correction step at $k=0$. This explains the small $e_{1_{k}}$ at $k=0$ compared to the initial error (see the zoomed-in plot in Fig. \ref{SMF_estimation_error_bounds}). 
%-----------------------------------------------------------------------------------------------------------------------------------------
\begin{figure}[!hbt]
\centering
\includegraphics[width=0.9\columnwidth]{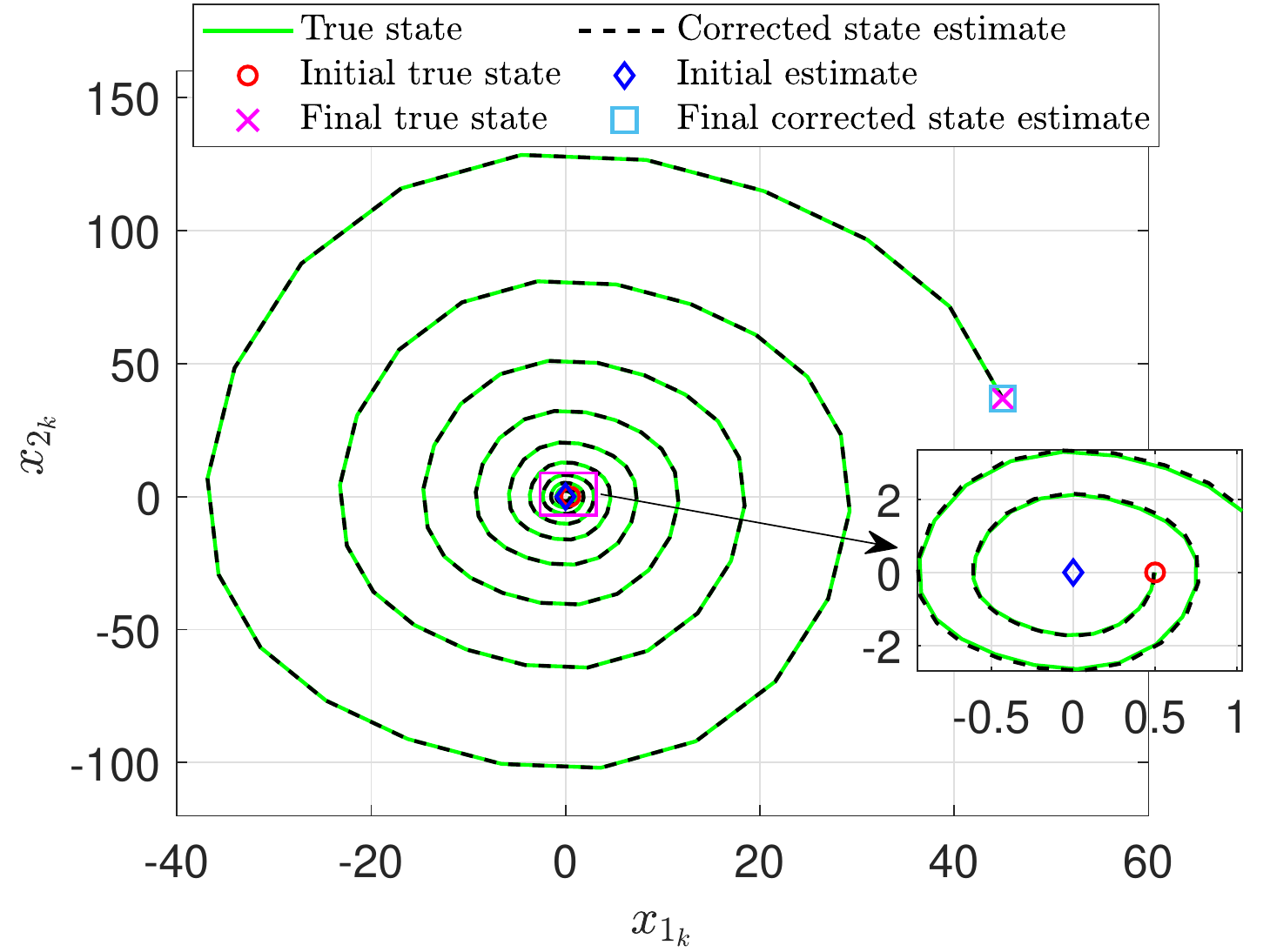}
\caption{The true state and corrected state estimate trajectories in the phase plane (Example-1).} 
\label{SMF_phase_plane}
\end{figure}
%----------------------------------------------------------------------------------------
\begin{table}[!hbt] 
\centering
\caption{Estimation error comparisons over $T = 201$ time-steps ($T_f = 200$)} \label{Table - Comparison}
\begin{tabular}{|c|c|c|} 
\hline
Item & Proposed SMF & Balandin et al. (Ref. \cite{Balandin_et_al_Automatica_2020}) \\ \hline
$\frac{1}{T} \sum |\bm{e}_k| $ & 0.0434 & 0.0477 \\ \hline
$\frac{1}{T} \sum e^2_{1_{k}}$ & 0.0002 & 0.0015 \\ \hline
$\frac{1}{T} \sum e^2_{2_{k}}$ & 0.00267 & 0.0028 \\ \hline
\end{tabular}
\end{table}
%--------------------------------------------------------------------------------------

Next, we compare the results of the proposed SMF with the ones corresponding to the SMF framework given in Balandin et al. (Ref. \cite{Balandin_et_al_Automatica_2020}, the discrete version). The system parameters, $\Delta t$, disturbances, and initial conditons considered for both the frameworks are as shown in Eq. \eqref{Parameter values for the SMF}. However, the disturbance ellipsoid shape matrices for the framework in Ref. \cite{Balandin_et_al_Automatica_2020} are adopted from the example in section 3 in Ref. \cite{Balandin_et_al_Automatica_2020}. Whereas, for the proposed SMF, the disturbance ellipsoid shape matrices are as shown in Eq. \eqref{Parameter values for the SMF}. The difference in the disturbance ellipsoid shape matrices are due to the different kinds of ellipsoidal assumptions utilized in this paper, compared to the ones in Ref. \cite{Balandin_et_al_Automatica_2020}. However, these values are chosen such that both sets of disturbance ellipsoid shape matrices are equivalent. The comparison results are shown in Figs. \ref{SMF_comparison_est_error}, \ref{SMF_comparison_trace}. Fig. \ref{SMF_comparison_est_error} shows the comparison in estimation error norms % the norm of estimation errors
where estimation errors at the correction steps for the proposed SMF are depicted. The proposed SMF is able to reduce the error norm at $k=0$ due to the initial correction step (see the zoomed-in plot in Fig. \ref{SMF_comparison_est_error}). After that, both the SMFs have qualitatively similar error norms. Table \ref{Table - Comparison} illustrates quantitative comparisons of the estimation errors. Clearly, the proposed SMF outperforms the SMF in Ref. \cite{Balandin_et_al_Automatica_2020} in terms of the mean absolute error ($\frac{1}{T} \sum |\bm{e}_k| $) and mean squared errors ($\frac{1}{T} \sum e^2_{1_{k}}$, $\frac{1}{T} \sum e^2_{2_{k}}$). % the mean absolute error ($\frac{1}{T_f} \sum |\bm{e}_k| $) and mean squared errors ($\frac{1}{T_f} \sum e^2_{1_{k}}$, $\frac{1}{T_f} \sum e^2_{2_{k}}$).
%%%%%%%%%%%%%%%%%%%%%%%%%%%%%%%%%%%%%%%%%%%%%%%%%%%%%%%%%%%%%%%%
\begin{figure}[!hbt]
\centering
\includegraphics[width=0.8\columnwidth]{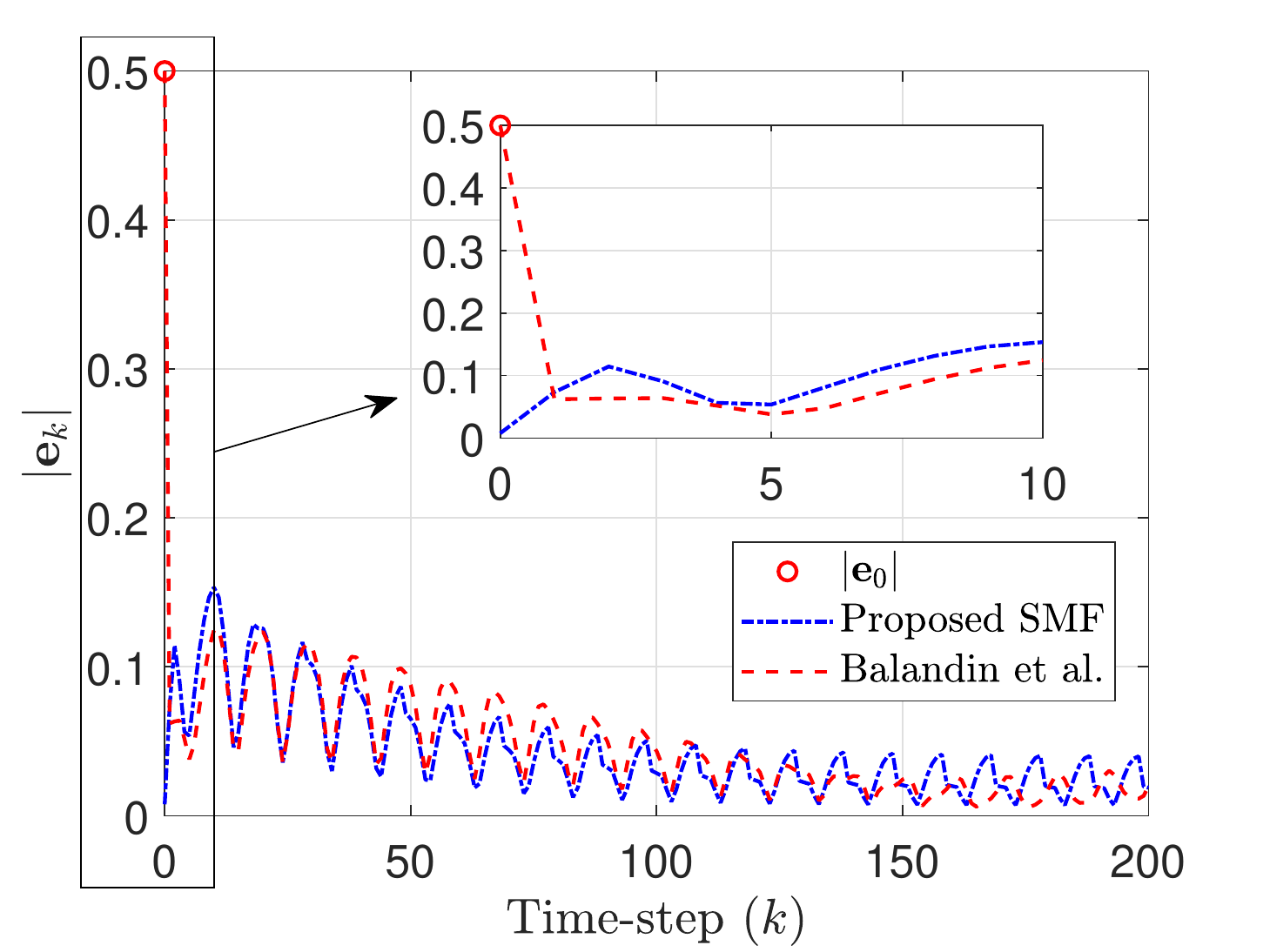}
\caption{Comparison of the estimation error norms where $|\bm{e}_0| = |\bm{x}_0 - \hat{\bm{x}}_0|$ (Example-1).} 
\label{SMF_comparison_est_error}
\end{figure}
%-----------------------------------------------------------------------------------------------------------------------------------------
\begin{figure}[!hbt]
\centering
\includegraphics[width=0.8\columnwidth]{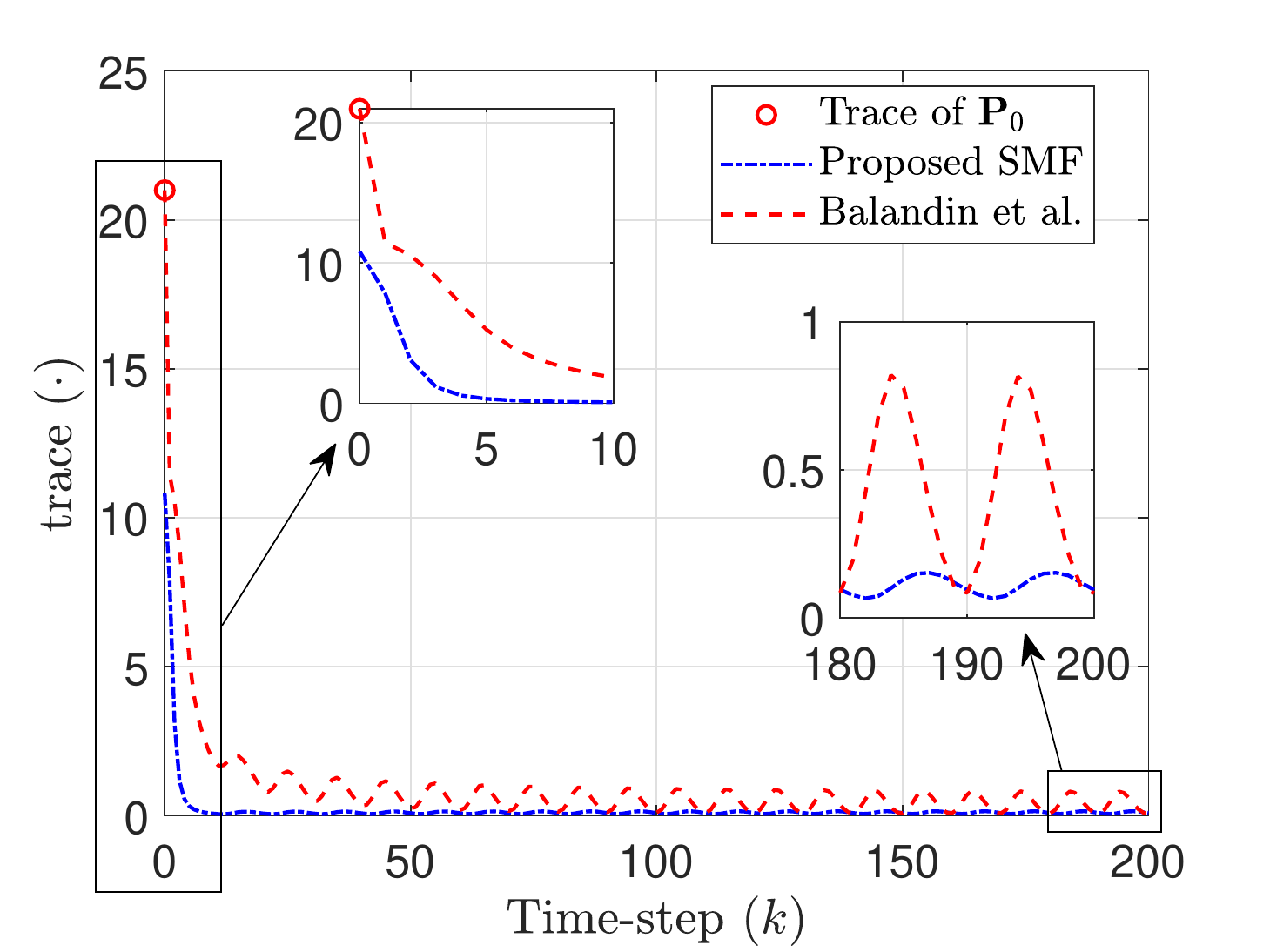}
\caption{Comparison of the trace of the ellipsoid shape matrices (Example-1).} 
\label{SMF_comparison_trace}
\end{figure}

The trace of the shape matrices, corresponding to the correction ellispoids of the proposed SMF and the estimation ellipsoids of the SMF in Ref. \cite{Balandin_et_al_Automatica_2020}, are shown in Fig. \ref{SMF_comparison_trace}. These results show that overall the correction ellipsoids of the proposed SMF are smaller in `size' compared to the estimation ellipsoids of the SMF in Ref. \cite{Balandin_et_al_Automatica_2020}. Thus, the error bounds shown in Fig. \ref{SMF_estimation_error_bounds} for the proposed SMF are tighter compared to the ones for the SMF in Ref. \cite{Balandin_et_al_Automatica_2020}. Also, the proposed SMF is able to reduce the `size' of the correction ellipsoid at $k=0$ due to the initial correction step, as shown in the zoomed-in plot in Fig. \ref{SMF_comparison_trace}. Note that the proposed SMF employs a two-step filtering approach wherein two SDPs are solved during every filter recursion which results in optimal (minimum trace) correction and prediction ellipsoids (with the respective state estimates at the centers). On the other hand, the SMF in Ref. \cite{Balandin_et_al_Automatica_2020} employs a one-step filtering technique with a combined correction and prediction step. Thus, the optimization process for estimation happens only once during each recursion of the SMF in \cite{Balandin_et_al_Automatica_2020}. This is likely the reason for better overall performance of the proposed SMF in this example. 
%%%%%%%%%%%%%%%%%%%%%%%%%%%%%%%%%%%%%%%%%%%%%%%%%%%%%%%%%%%%%%%%%%%%%%%%%%%%%%%%%%%%%%%%%%%%%%%%%%%%%%%%%%%%%%%%%%%%%%%%%%%%%%%%%%%%%%%%%%%%%%%%%%%%%%%%%%%%%%%%%%%%%%%%%%%%%%%%%%%%%%%%%%%%%%%%%%%%%%%%%%%%%%%%%%%%%%%%%%%%%%%%%%%%%%%%%%%%%%%%%%%%%%%%%%%%%%%%%%%%%%%%%%%%%%%%%%%%%%%%%%%%
\vspace{-0.2cm}
\subsection{Example-2}
In this example, we illustrate results of the proposed SMF-based leader-follower synchronization protocol. We consider four agents, i.e., $N=4$. Matrices related to the dynamics of the leader and the agents are
\begin{equation}
\bm{A} = \begin{bmatrix}
0 & -1 \\
1 & 0
\end{bmatrix}, \ \bm{B} = \bm{I}_2, \ \bm{C} = [1 \quad 0], \ \bm{D} = 1, \ \bm{G} = \bm{I}_2
\end{equation}
%----------------------------------------------------------------------------------------------------------------------------------
\begin{figure}[!hbt]
\centering
\subfigure[]{\includegraphics[width=0.375\columnwidth]{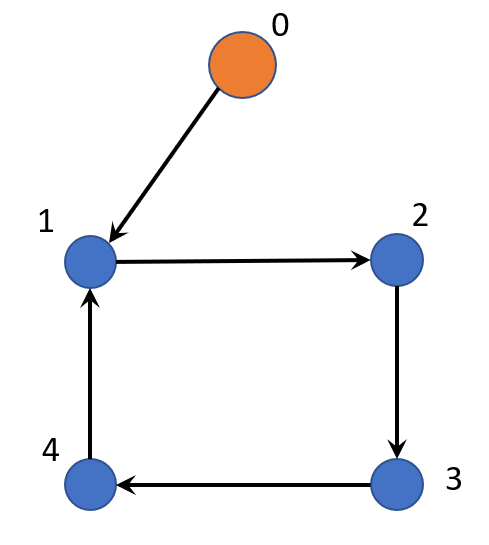} \label{Communication_topology}}
\subfigure[]{\includegraphics[width=0.575\columnwidth]{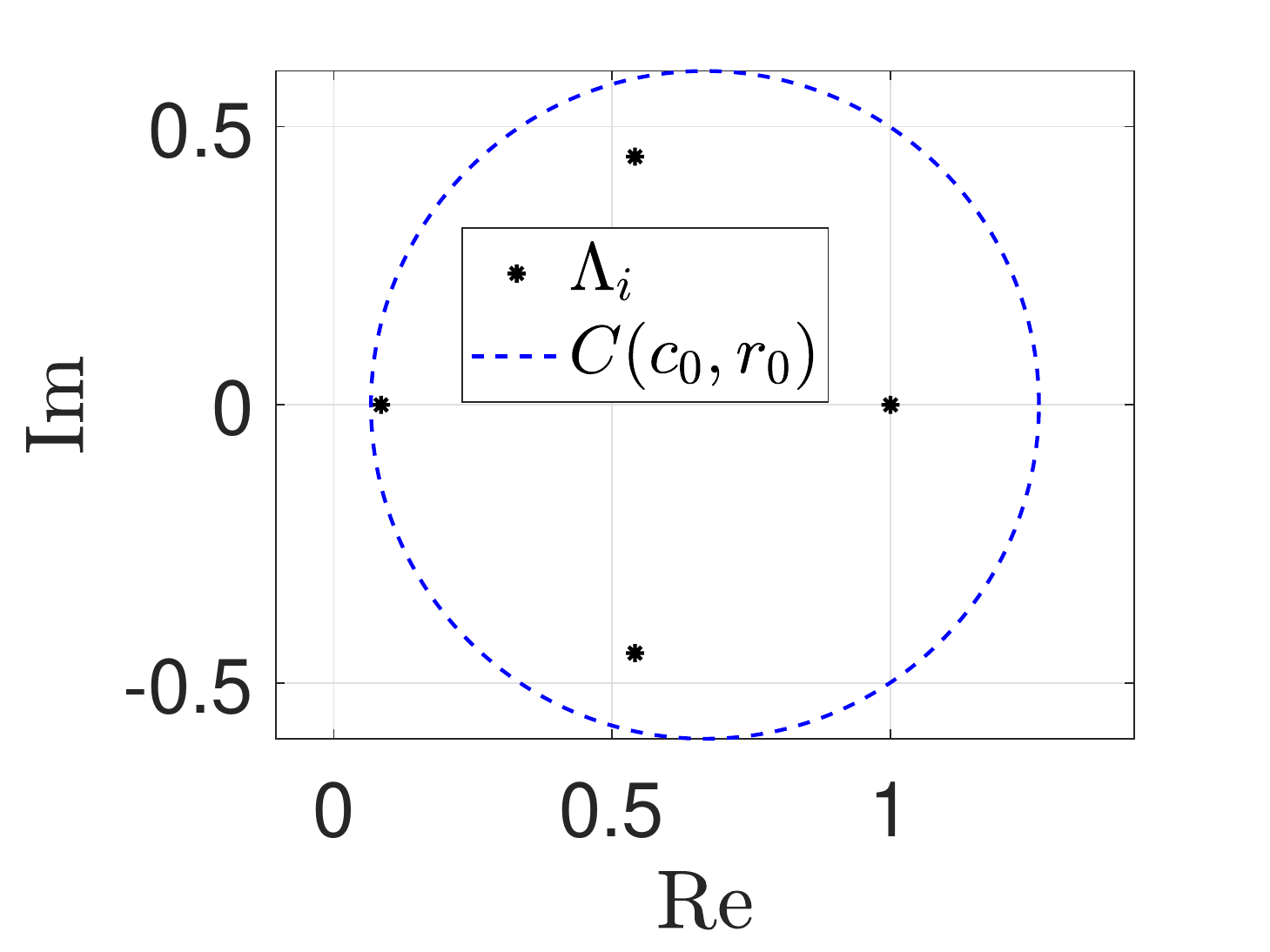} \label{Graph_related_plot}}
\caption{(a) The interaction graph and (b) eigenvalues of $\bm{\Gamma}$ ($\Lambda_i, \ i=1,2,3,4$) in the complex plane with $C(c_0,r_0)$ (Example-2).} 
\label{Communication_topology and Graph_related_plot}
\end{figure}
%---------------------------------------------------------------------------------------------------------------------------
where $\bm{A}$ is marginally stable. Also, Assumption \ref{Assumptions on A and B matrices} is satisfied with the above choices of $\bm{A}$ and $\bm{B}$. Ellipsoidal parameters related to the SMFs of the agents are $\bm{P}^{(i)}_0 = 2 \bm{I}_{2}, \bm{Q}^{(i)}_k = 0.1 \bm{I}_{2}, \ \bm{R}^{(i)}_k = 0.1$ for $i=1,2,3,4$. The initial state estimates of the agents are as follows: $\hat{\bm{x}}^{(1)}_0 = [50 \quad -50]^\text{T}$, $\hat{\bm{x}}^{(2)}_0 = \hat{\bm{x}}^{(1)}_0 $, $\hat{\bm{x}}^{(3)}_0 = [-50 \quad 50]^\text{T}$, $\hat{\bm{x}}^{(4)}_0 = \hat{\bm{x}}^{(3)}_0 $. The true initial state for the agents 1 and 2 ($\bm{x}^{(1)}_0$, $\bm{x}^{(2)}_0$) are chosen randomly (uniform distribution) between $[50 \quad -50]^\text{T}$ and $[51 \quad -49]^\text{T}$. Similarly, the true initial state for the agents 3 and 4 ($\bm{x}^{(3)}_0$, $\bm{x}^{(4)}_0$) are chosen randomly (uniform distribution) between $[-50 \quad 50]^\text{T}$ and $[-49 \quad 51]^\text{T}$. The input disturbances ($\bm{w}^{(i)}_k$, $i=1,2,3,4$) are chosen randomly (uniform distribution) between $-0.05 \bm{1}_{2}$ and $0.05 \bm{1}_{2}$, and the output disturbances ($\bm{v}^{(i)}_k$, $i=1,2,3,4$) are chosen randomly (uniform distribution) between $-0.05$ and $0.05$. Thus, Assumption \ref{Assumption 2} has been satisfied with the above  parameters, initial conditions, and disturbance terms. The initial state of the leader is chosen as $\bm{x}^{(0)}_0 = [5 \ -5]^\text{T}$.

The interaction graph is shown in Fig. \ref{Communication_topology} and Assumption \ref{Assumption on the graph structure} holds for this interaction graph. Thus, we have
\begin{equation}
\begin{split}
\pmb{\mathscr{L}} &= \begin{bmatrix}
1 & 0 & 0 & -1 \\
-1 & 1 & 0 & 0 \\
0 & -1 & 1 & 0 \\
0 & 0 & -1 & 1
\end{bmatrix},
\end{split}
\end{equation}
$\pmb{\mathcal{G}} = \text{diag}(1,0,0,0), \pmb{\mathscr{D}} = \text{diag}(1,1,1,1)$. With regards to Lemma \ref{Lemma: Choices of c and K for A_c Schur}, we choose $\pmb{\mathcal{Q}} = 0.1 \bm{I}_2$, $c_0 = (2/3)$, $r_0 = 0.6$. Clearly, $C(c_0,r_0)$ contains all the eigenvalues of $\bm{\Gamma}$ as shown in Fig. \ref{Graph_related_plot}. Also, we have $r= 1$ and $(r_0/c_0) = 0.9 < r$. Hence, the conditions for Lemma \ref{Lemma: Choices of c and K for A_c Schur} are satisfied and we take $c= (1/c_0) = 1.5$, $\bm{K} = (\bm{B}^\text{T} \pmb{\mathcal{P}} \bm{B})^{-1} \bm{B}^\text{T} \pmb{\mathcal{P}} \bm{A}$.

The synchronization results are shown in Figs. \ref{Synchronization time history} and \ref{Global_disagreement_error_norm}. Figure \ref{Synchronization time history} shows that the trajectories of the agents converge close to that of the leader. As a result, the normalized global disagreement error norm converges to a neighborhood of zero (Fig. \ref{Global_disagreement_error_norm}). The red dotted line in Fig. \ref{Global_disagreement_error_norm} denotes the conservative bound in Eq. \eqref{final bound of |delta_k|} for which we have utilized $p_0 = 2$ and $\bar{q} = 0.1$. Also, for $\bar{\mu}$, we have taken $\alpha = 1.1$ and $\mu = 0.9$. For this choice of $\alpha$ and $\mu$, $|\bm{A}_c^k| \leq \alpha \mu^k$ is satisfied, as shown in Fig. \ref{Norm_Ac_upper_bound}. With the above values, the conservative upper bound is equal to 2.462, which is shown using the red dotted line in Fig. \ref{Global_disagreement_error_norm}.

%----------------------------------------------------------------------------------------------------------------------------------
\begin{figure}[!hbt]
\centering
\subfigure[True states of the leader and the agents]{\includegraphics[width=1\columnwidth]{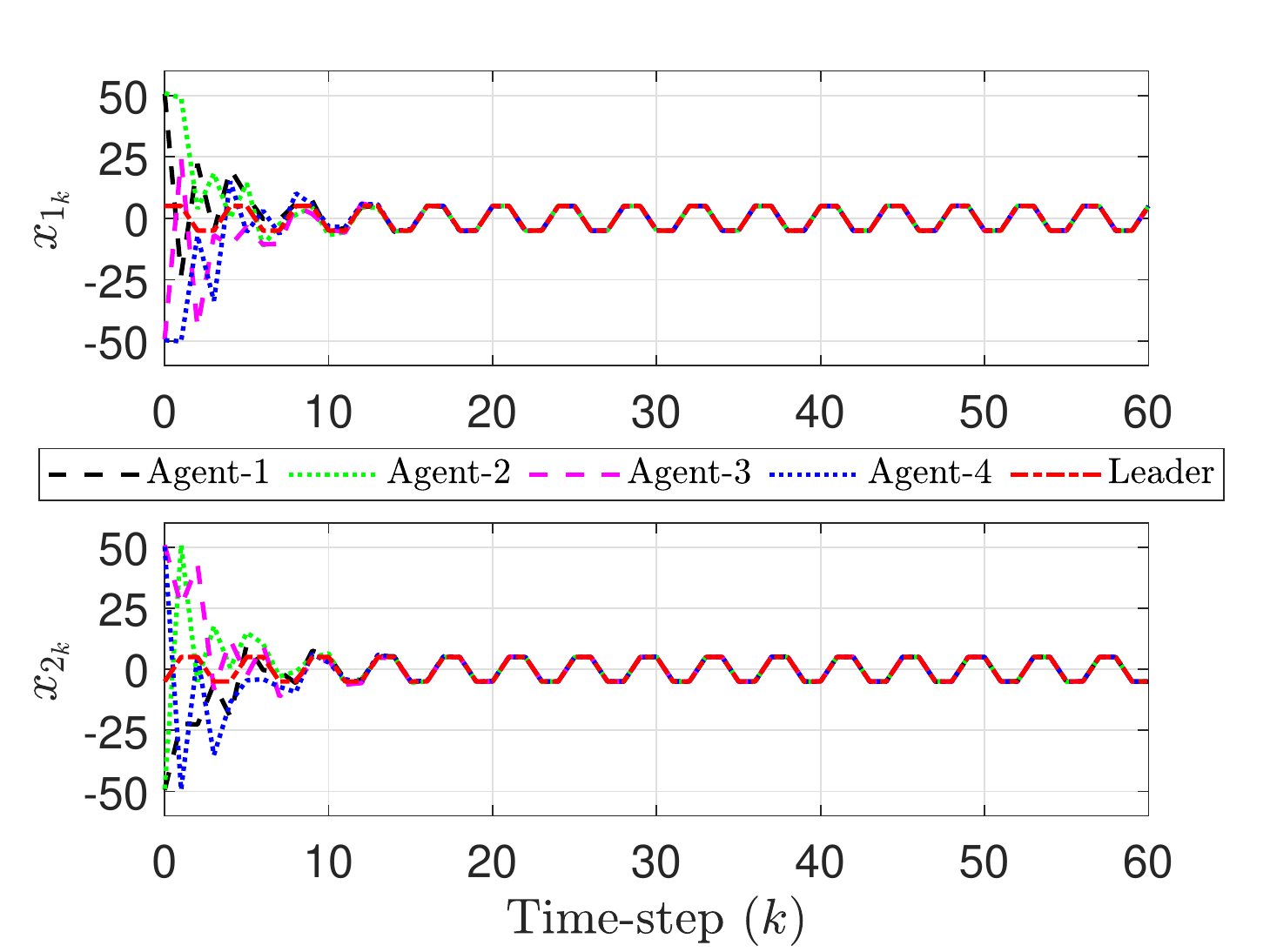} \label{Synchronization time history}}
\subfigure[Normalized global disagreement error norm]{\includegraphics[width=0.65\columnwidth]{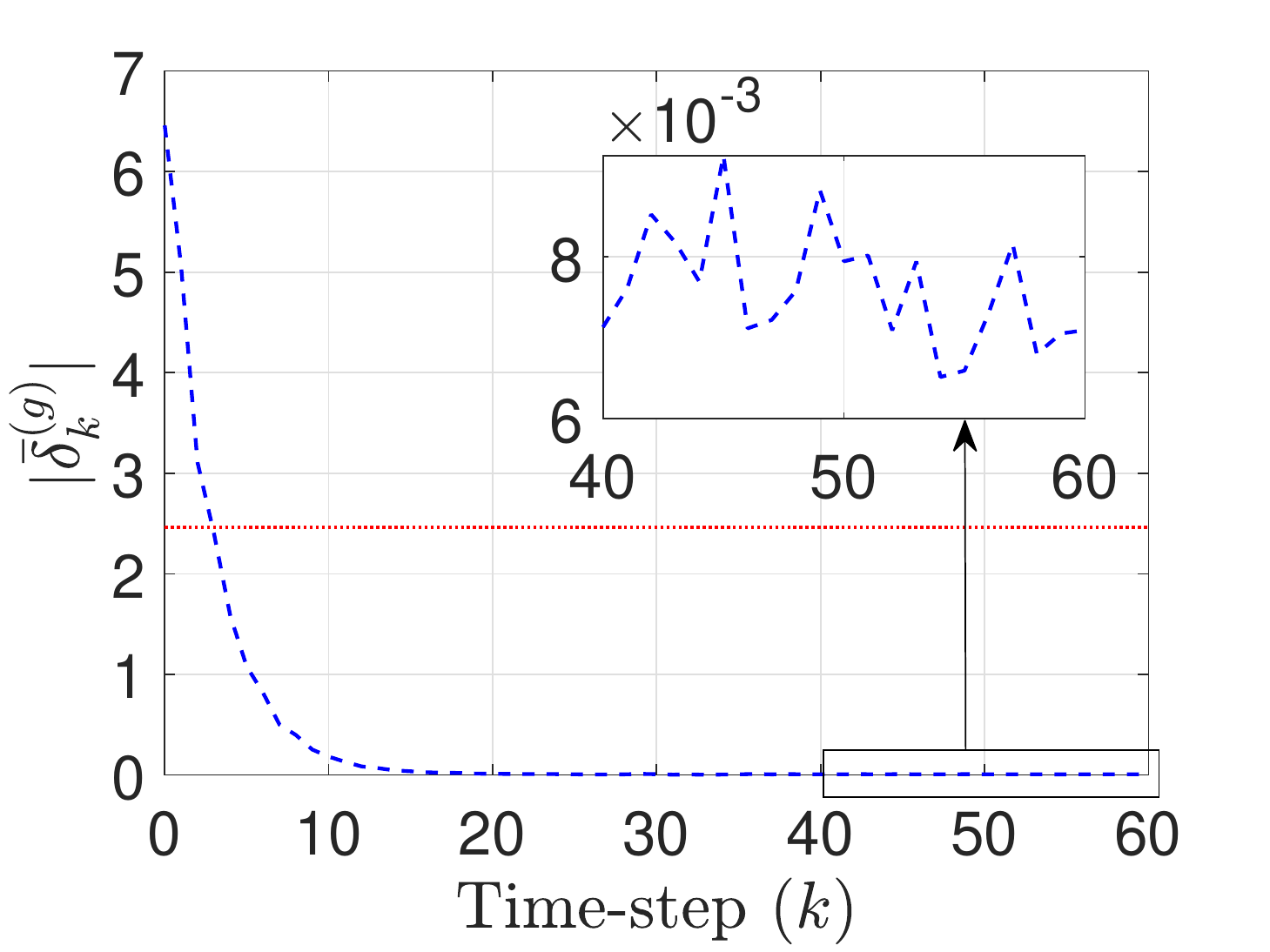} \label{Global_disagreement_error_norm}}
\subfigure[$|\bm{A}^k_c|$ and the upper bound]{\includegraphics[width=0.65\columnwidth]{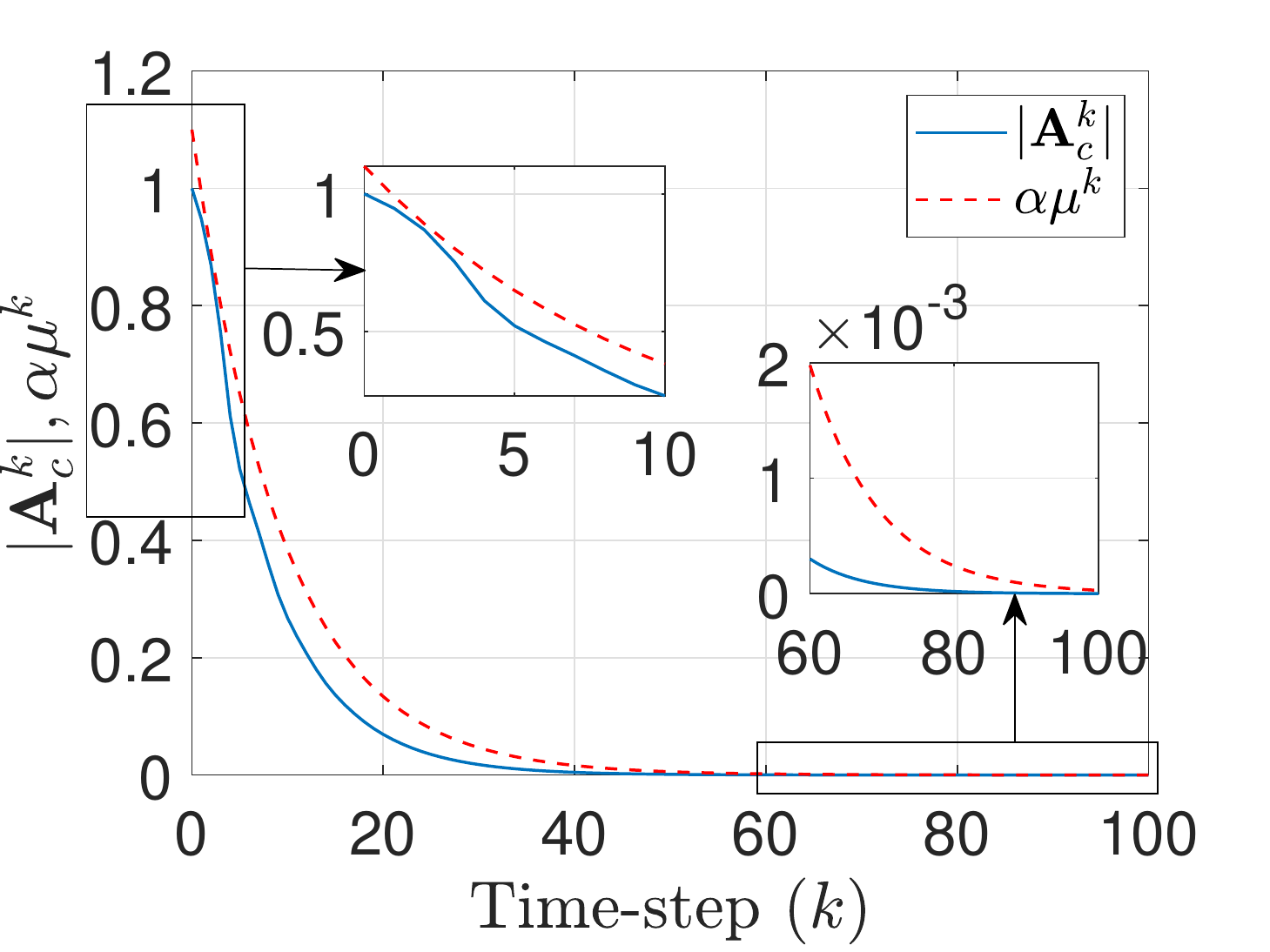} \label{Norm_Ac_upper_bound}}
\caption{Simulation results for the Example-2.} 
\label{Consensus_two_step_SMF}
\end{figure}
%-----------------------------------------------------------------------------------------------------------------------------------------

The estimation results corresponding to the SMFs of the agents are shown in Figs. \ref{Estimation_two_step_SMF_1}, \ref{Estimation_two_step_SMF_2}, \ref{Estimation_error_norm_agents}, \ref{Traces_of_correction_ellipsoids_agents}. The estimation errors (at the correction steps) for agent $i$'s SMF are denoted by $\bm{e}^{(i)}_k = [e^{(i)}_{1_{k}} \quad e^{(i)}_{2_{k}}]^\text{T}$ and the initial errors are denoted by $\bm{e}^{(i)}_0 = \bm{x}^{(i)}_0 - \hat{\bm{x}}^{(i)}_0$ ($i=1,2,3,4$). Figs. \ref{Estimation_two_step_SMF_1}, \ref{Estimation_two_step_SMF_2} show that the SMFs for the agents perform adequately as the estimation errors remain in a neighborhood of zero and the error bounds decrease from the respective initial values. Also, the estimation errors are contained within the error bounds which mean the SMFs of the agents are able to contain the respective true states inside the respective correction ellipsoids. The estimation error norms, shown in Fig. \ref{Estimation_error_norm_agents}, further illustrate the effectiveness of the SMFs and demonstrate that the SMFs are able to reduce the estimation errors from the respective initial values, starting from the correction step at $k=0$. The results in Fig. \ref{Estimation_error_norm_agents} essentially verify our earlier use of $||\bm{e}^{(g)}|| \leq |\bm{e}^{(g)}_{0}|$ in deriving the conservative bound in Eq. \eqref{final bound of |delta_k|}. 
%-----------------------------------------------------------------------------------------------------------------------------------------
\begin{figure}[!hbt]
\centering
\includegraphics[width=1\columnwidth]{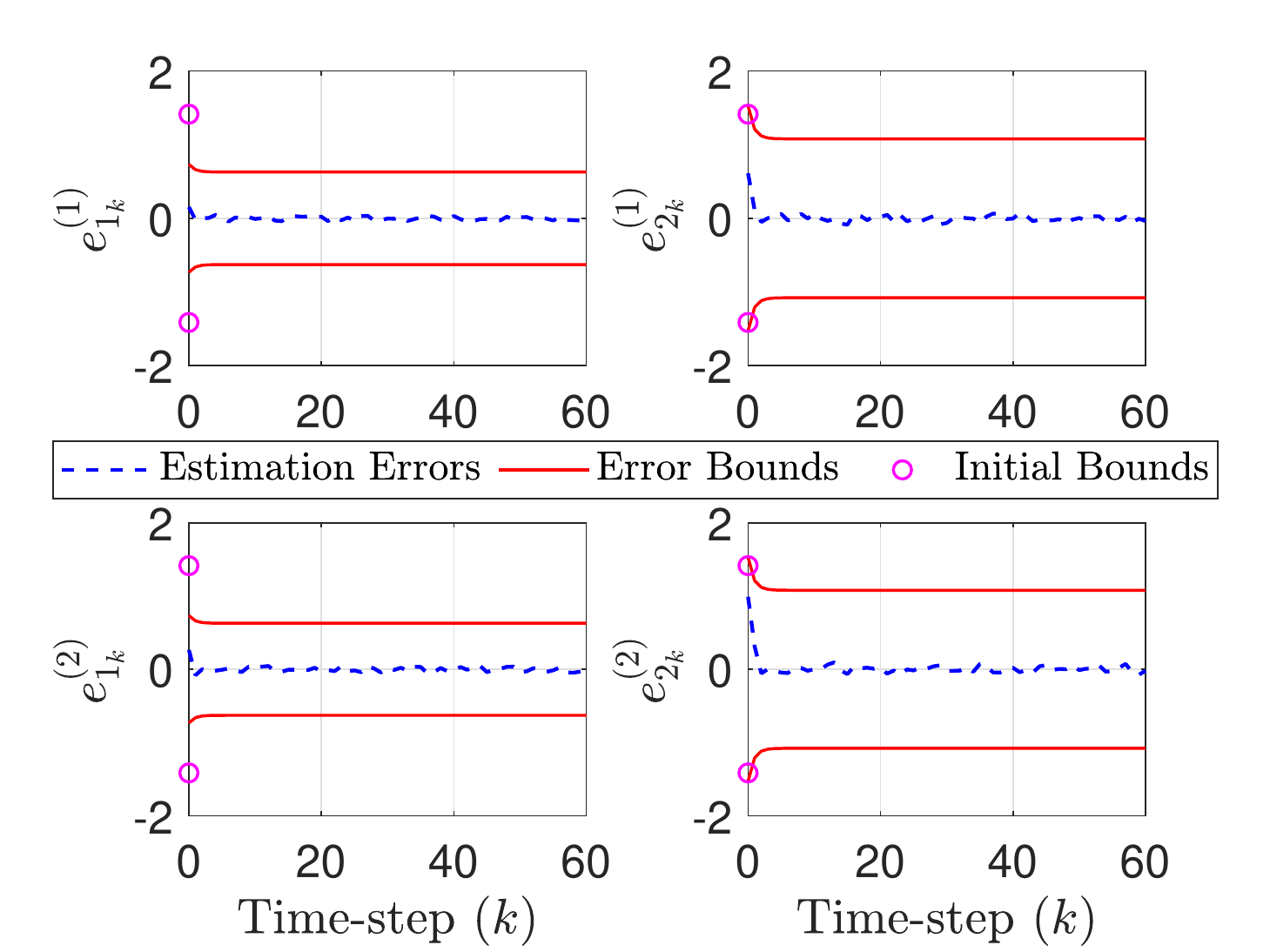}
\caption{Estimation results for SMFs of agents 1 and 2 (Example-2).} 
\label{Estimation_two_step_SMF_1}
\end{figure}
%-----------------------------------------------------------------------------------------------------------------------------------------
\begin{figure}[!hbt]
\centering
\includegraphics[width=1\columnwidth]{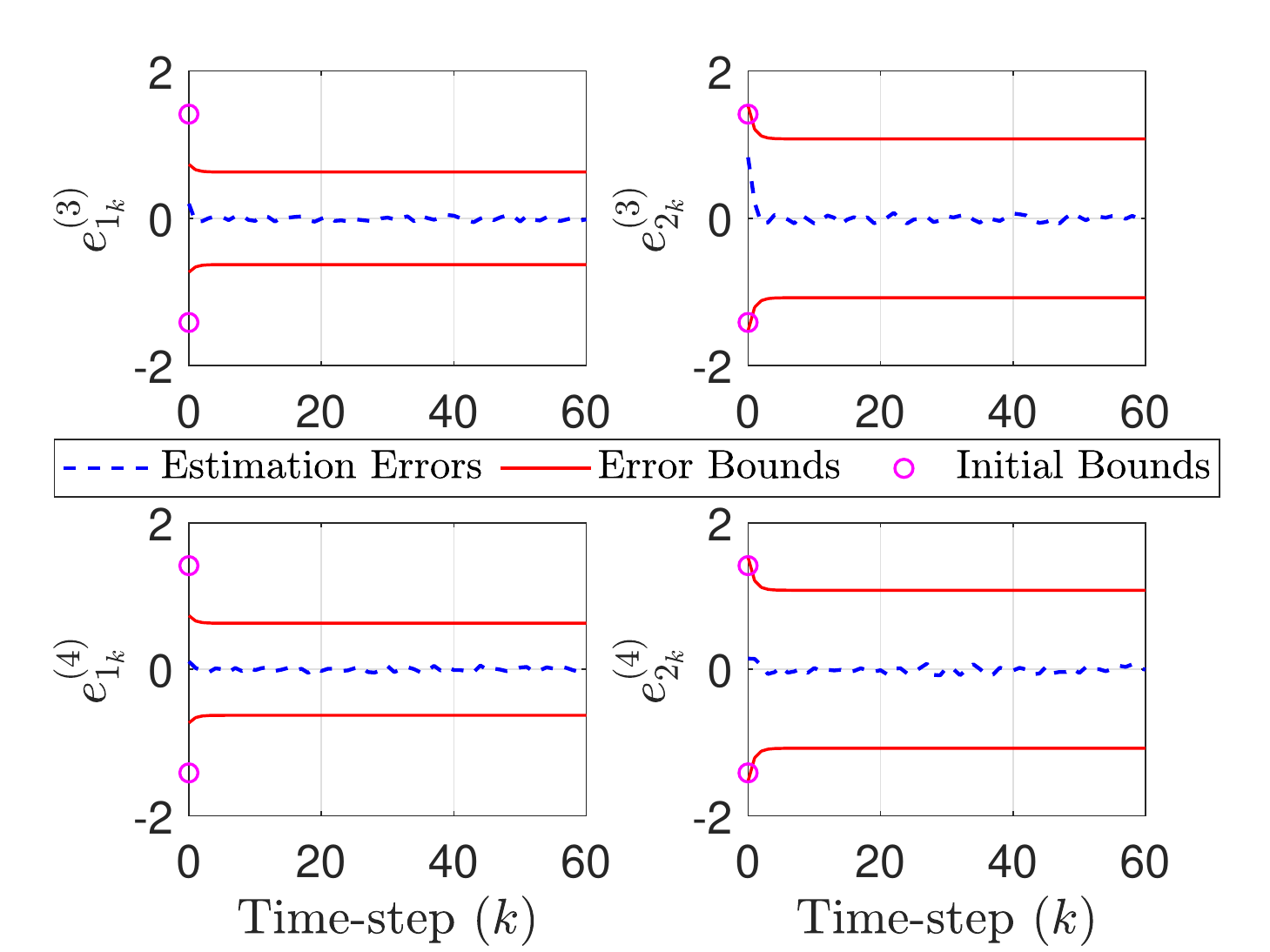}
\caption{Estimation results for SMFs of agents 3 and 4 (Example-2).} 
\label{Estimation_two_step_SMF_2}
\end{figure} 
%--------------------------------------------------------------------------------------------------------------------------------------------
The trace of correction ellipsoid shape matrices for the SMFs of the agents are shown in Fig. \ref{Traces_of_correction_ellipsoids_agents} where $\bm{P}^{(i)}_{k|k}$ ($i=1,2,3,4$) denote the shape matrices of agent $i$'s correction ellipsoids. Clearly, SMFs of the agents are able to reduce the trace from the initial values and construct optimal (minimum trace) correction ellipsoids at each time-step (starting from $k=0$). Quantitatively, the trace of these shape matrices converge approximately to 1.5 (see Fig. \ref{Traces_of_correction_ellipsoids_agents}), which is approximately a 2.667-fold decrease with respect to the initial trace of 4. The trends shown in Fig. \ref{Traces_of_correction_ellipsoids_agents} for all the agents are roughly the same as the same set of ellipsoidal parameters is utilized for the SMFs of all the agents and the agents have identical dynamics.
%-----------------------------------------------------------------------------------------------------------------------------------------
\begin{figure}[!hbt]
\centering
\includegraphics[width=1\columnwidth]{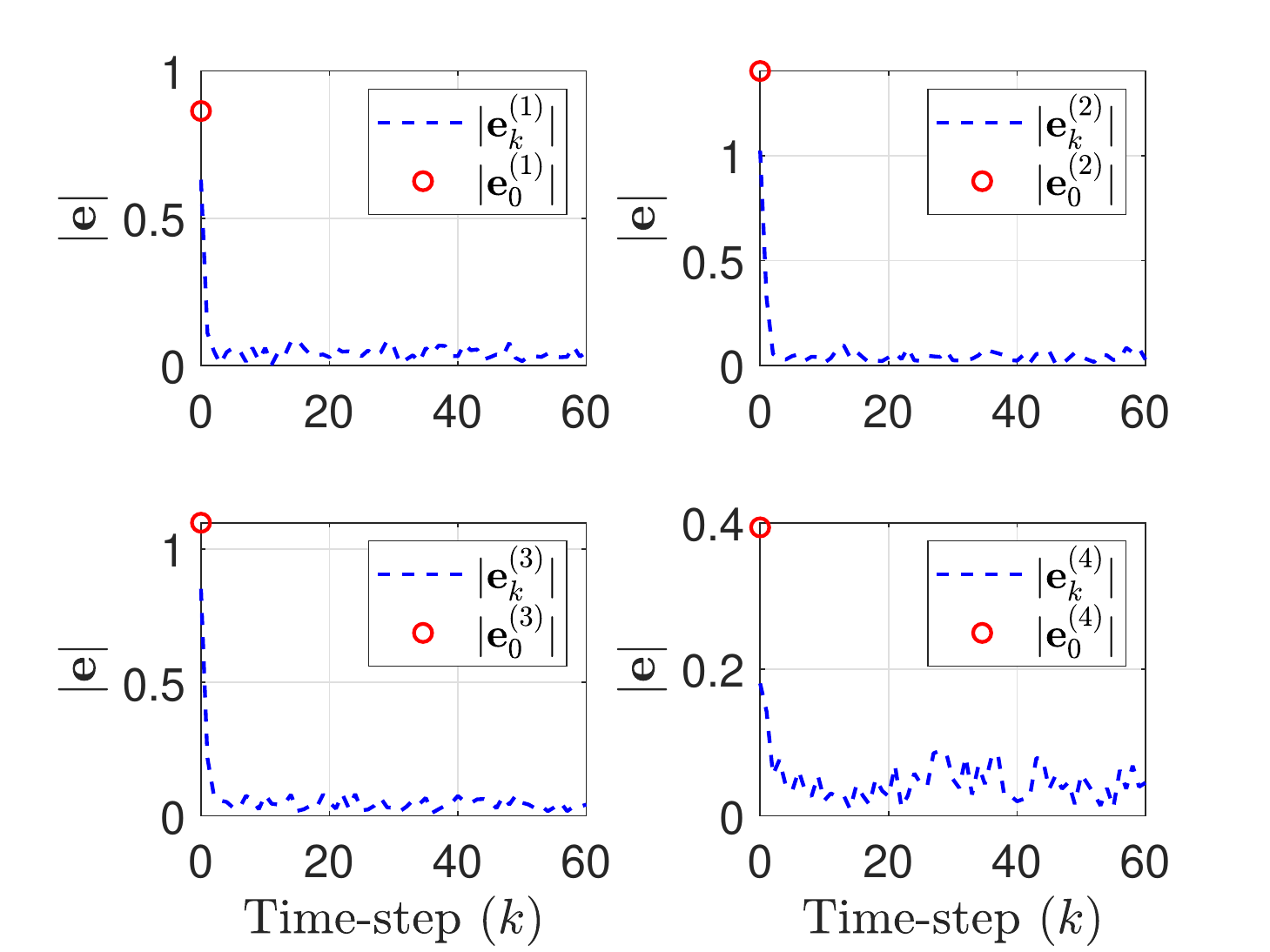}
\caption{Estimation error norms for SMFs of the agents (Example-2).} 
\label{Estimation_error_norm_agents}
\end{figure}
%-----------------------------------------------------------------------------------------------------------------------------------------
\begin{figure}[!hbt]
\centering
\includegraphics[width=1\columnwidth]{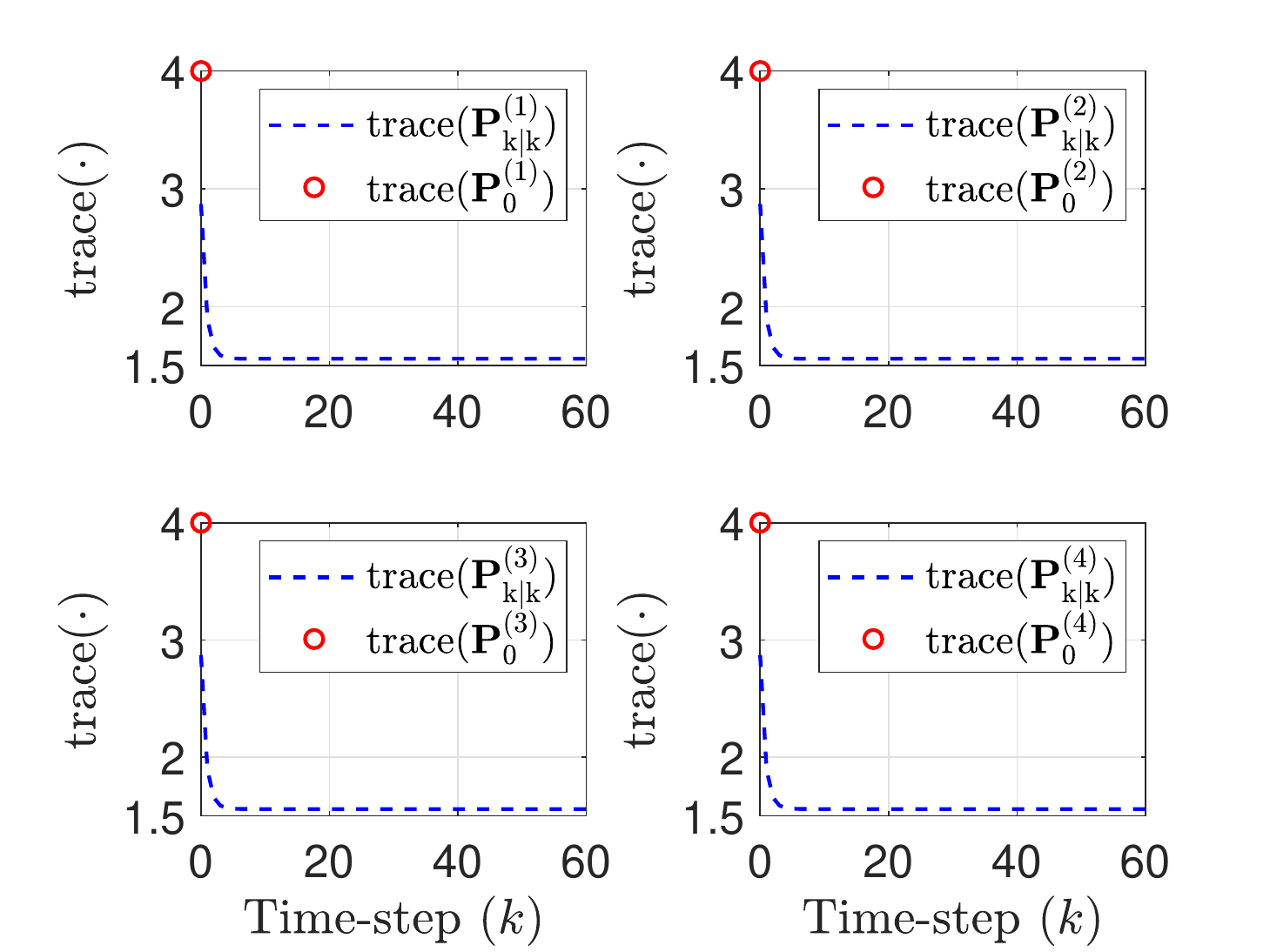}
\caption{Trace of correction ellipsoid shape matrices for SMFs of the agents (Example-2).} 
\label{Traces_of_correction_ellipsoids_agents}
\end{figure}
%-------------------------------------------------------------------------------------------------------------------------------------------

Finally, consider this example with different values of $\bm{w}^{(i)}_k$, $\bm{v}^{(i)}_k$, $\bm{Q}^{(i)}_k$, $\bm{R}^{(i)}_k$ ($i=1,2,3,4$) while keeping all other conditions and parameters unchanged. Now, let us allow for higher magnitudes of disturbances (with $\bm{Q}^{(i)}_k$, $\bm{R}^{(i)}_k$ properly chosen such that Assumption \ref{Assumption 2 - process and measurement noise ellipsoids} is satisfied) and compare $|\bar{\bm{\delta}}_k|$ results with the one given in Fig. \ref{Global_disagreement_error_norm}. Results of this study are given in Table \ref{Table - Global Error Norm Comparison} where the following two comparison metrics are used: (i) $\frac{1}{T} \sum_{k=0}^{T_f} |\bar{\bm{\delta}}_k|$: mean value of $|\bar{\bm{\delta}}_k|$; (ii) $\sqrt{\frac{1}{T} \sum_{k=0}^{T_f} |\bar{\bm{\delta}}_k|^2}$ : root mean square value of $|\bar{\bm{\delta}}_k|$. Also, $\bm{w}^{(i)}_k$ are chosen randomly (uniform distribution) between $-\alpha_w \bm{1}_{2}$ and $\alpha_w \bm{1}_{2}$, and $\bm{v}^{(i)}_k$ are chosen randomly (uniform distribution) between $-\alpha_v$ and $\alpha_v$. Thus, the first row in Table \ref{Table - Global Error Norm Comparison} corresponds to the result in Fig. \ref{Global_disagreement_error_norm}.
%--------------------------------------------------------------------------------------------------------------------------------
\begin{table}[!hbt] 
\centering
\caption{$|\bar{\bm{\delta}}_k|$ comparisons over $T = 61$ time-steps ($T_f = 60$)} \label{Table - Global Error Norm Comparison}
% \begin{tabular}{|c|c|c|}
\begin{tabular}{| >{\centering\arraybackslash}m{3.5cm}|>{\centering\arraybackslash}m{1.5cm}| >{\centering\arraybackslash}m{2cm}|}
\hline
Disturbance parameters & $\frac{1}{T} \sum_{k=0}^{T_f} |\bar{\bm{\delta}}_k|$ & $ \sqrt{\frac{1}{T} \sum_{k=0}^{T_f} |\bar{\bm{\delta}}_k|^2} $ \\ \hline
%--------------------------------------------------------------------------------------------------------------------------------------------
$\alpha_w = \alpha_v = 0.05$, $\bm{Q}^{(i)}_k = 0.1 \bm{I}_{2}, \ \bm{R}^{(i)}_k = 0.1$ & 0.3706 & 1.1985 \\ \hline
%--------------------------------------------------------------------------------------------------------------------------------------------
$\alpha_w = \alpha_v = 0.5$, $\bm{Q}^{(i)}_k = \bm{I}_{2}, \ \bm{R}^{(i)}_k = 1$        & 0.4219 & 1.2052 \\ \hline
%--------------------------------------------------------------------------------------------------------------------------------------------
$\alpha_w = \alpha_v = 1$, $\bm{Q}^{(i)}_k = 2 \bm{I}_{2}, \ \bm{R}^{(i)}_k = 1$        & 0.4730 & 1.2124 \\ \hline
%--------------------------------------------------------------------------------------------------------------------------------------------
\end{tabular}
\end{table}
%------------------------------------------------------------------------------------------------------------------------------
We observe that both the metrics in Table \ref{Table - Global Error Norm Comparison} are comparable among the three cases studied, despite the higher magnitudes of disturbances considered for the two cases in second and third rows of Table \ref{Table - Global Error Norm Comparison}. Therefore, the $|\bar{\bm{\delta}}_k|$ trends for these two cases with higher disturbance magnitudes would be qualitatively similar to the one shown in Fig. \ref{Global_disagreement_error_norm}.
%%%%%%%%%%%%%%%%%%%%%%%%%%%%%%%%%%%%%%%%%%%%%%%%%%%%%%%%%%%%%%%%%%%%%%%%%%%%%%%%%%%%%%%%%%%%%%%%%%%%%%%%%%%%%%%%%%%%%%%%%%%%%%%%%%%%%%%%%%%%%%%%%%%%%%%%%%%%%%%%%%%%%%%%%%%%%%%%%%%%%%%%%%%%%%%%%%%%%%%%%%%%%%%%%%%%%%%%%%%%%%%%%%%%%%%%%%%%%%%%%%%%%%%%%%%%%%%%%%%%%%%%%%%%%%%%%%%%%%%%%%%%%%%%%%%%%%%%%%%%%%%%%%%%%%%%%%%%%%%%%%%%%%%%%%%%%%%%%%%%%%%%%%%%%%%%%%%%%%%%%%%%%%%%%%%%%%%%%%%%%%%%%%%%%%%%%%%%%%%%%%%%%%%%%%%%%%%%%%%%%%%%%
\vspace{-0.5cm}
\section{Conclusion} \label{Conclusion}
A set-membership filtering-based leader-follower synchronization protocol for high-order discrete-time linear multi-agent systems has been put forward for which the global error system is shown to be input-to-state stable with respect to the input disturbances and estimation errors. A monotonically decreasing upper bound on the norm of the global disagreement error vector is calculated. Our future work would involve extending the proposed formulation for discrete-time nonlinear dynamical systems and switching network topologies. Also, we would extend the results in this paper by considering a control input for the leader or the leader to be any bounded reference trajectory. 
%%%%%%%%%%%%%%%%%%%%%%%%%%%%%%%%%%%%%%%%%%%%%%%%%%%%%%%%%%%%%%%%%%%%%%%%%%%%%%%%%%%%%%%%%%%%%%%%%%%%%%%%%%%%%%%%%%%%%%%%%%%%%%%%%%%%%%%%%%%%%%%%%%%%%%%%%%%%%%%%%%%%%%%%%%%%%%%%%%%%%%%%%%%%%%%%%%%%%%%%%%%%%%%%%%%%%%%%%%%%%%%%%%%%%%%%%%%%%%%%%%%%%%%%%%%%%%%%%%%%%%%%%%%%%%%%%%%%%%%%%%%%%%%%%%%%%%%%%%%%%%%%%%%%%%%%%%%%%%%%%%%%%%%%%%%%%%%%%%%%%%%%%%%%%%%%%%%%%%%%%%%%%%%%%%%%%%%%%%%%%%%%%%%%%%%%%%%%%%%%%%%%%%%%%%%%%%%%%%%%%%%%%
\vspace{-0.5cm} 
\begin{acknowledgment}
This research was supported by the Office of Naval Research under Grant No. N00014-18-1-2215.
\end{acknowledgment}
%%%%%%%%%%%%%%%%%%%%%%%%%%%%%%%%%%%%%%%%%%%%%%%%%%%%%%%%%%%%%%%%%%%%%%%%%%%%%%%%%%%%%%%%%%%%%%%%%%%%%%%%%%%%%%%%%%%%%%%%%%%%%%%%%%%%%%%%%%%%%%%%%%%%%%%%%%%%%%%%%%%%%%%%%%%%%%%%%%%%%%%%%%%%%%%%%%%%%%%%%%%%%%%%%%%%%%%%%%%%%%%%%%%%%%%%%%%%%%%%%%%%%%%%%%%%%%%%%%%%%%%%%%%%%%%%%%%%%%%%%%%%%%%%%%%%%%%%%%%%%%%%%%%%%%%%%%%%%%%%%%%%%%%%%%%%%%%%%%%%%%%%%%%%%%%%%%%%%%%%%%%%%%%%%%%%%%%%%%%%%%%%%%%%%%%%%%%%%%%%%%%%%%%%%%%%%%%%%%%%%%%%%
\vspace{-0.75cm}
\scriptsize{\bibliography{Bibliography}
\bibliographystyle{asmems4}}

\end{document}